\documentclass[12pt]{svmult}

\def\bibsection{\section*{References}}

\usepackage{type1cm}        % activate if the above 3 fonts are
\usepackage{makeidx}         % allows index generation
\usepackage{graphicx}        % standard LaTeX graphics tool
\usepackage{multicol}        % used for the two-column index
\usepackage[bottom]{footmisc}% places footnotes at page bottom
\usepackage{amsfonts}
\usepackage{comment}

\usepackage[margin=1in]{geometry}
\usepackage{newtxtext}       % 
\usepackage[varvw]{newtxmath}       % selects Times Roman as basic font
\usepackage{scalerel}
\usepackage{bbm}
\usepackage{dsfont}
\usepackage{bbm, dsfont}
\usepackage{mathrsfs}
\usepackage{mathtools}
\usepackage{commath}
\usepackage[ruled,vlined]{algorithm2e}
\usepackage{algpseudocode}
\usepackage{gensymb}

\usepackage{setspace} % Add the setspace package
\newcommand{\Ali}{\textcolor{black}}
\newcommand{\payam}[1]{#1}

\newcommand{\pr}{\mathds{P}}
\newcommand{\EE}{\mathds{E}}
\newcommand{\R}{{\mathbb{R}}}
\newcommand{\cg}{\textnormal{\textsl{g}}}
\newtheorem{assumption}[theorem]{Assumption}
\makeindex             % used for the subject index
\begin{document}

\title*{Data-driven verification and synthesis of stochastic systems via barrier certificates}
\author{Ali Salamati, Abolfazl Lavaei, Sadegh Soudjani, and Majid Zamani}
\institute{Ali Salamati \at Computer Science Department, Ludwig-Maximilians-Universit\"{a}t M\"{u}nchen, Germany, \email{ali.salamati@lmu.de}
\and Abolfazl Lavaei \at Institute for Dynamic Systems and Control, ETH, Zurich, Switzerland. \email{alavaei@ethz.ch}
\and Sadegh Soudjani \at School of Computing, Newcastle University, Newcastle, United Kingdom. \email{sadegh.soudjani@newcastle.ac.uk}
\and Majid Zamani \at Computer Science Department, University of Colorado Boulder, the USA and Computer Science Department, Ludwig-Maximilians-Universität München, Germany. \email{majid.zamani@colorado.edu}
}

\maketitle

\abstract{In this work, we study verification and synthesis problems for safety specifications over unknown discrete-time stochastic systems.
	When a model of the system is available, barrier certificates have been successfully applied for ensuring the satisfaction of safety
	specifications. In this work, we formulate the computation of barrier certificates as a robust convex program (RCP). Solving the
	acquired RCP is hard in general because the model of the system that appears in one of the constraints of the RCP is unknown. We
	propose a data-driven approach that replaces the uncountable number of constraints in the RCP with a finite number of constraints
	by taking finitely many random samples from the trajectories of the system. We thus replace the original RCP with a scenario
	convex program (SCP) and show how to relate their optimizers. We guarantee that the solution of the SCP is a solution of the
	RCP with a priori guaranteed confidence when the number of samples is larger than a specific value. This provides a lower
	bound on the safety probability of the original unknown system together with a controller in the case of synthesis. We also discuss
	an extension of our verification approach to a case where the associated robust program is non-convex and show how a similar
	methodology can be applied. Finally, the applicability of our proposed approach is illustrated through three case studies.}


\section{Introduction}
Ensuring safety and temporal requirements on cyber-physical systems is becoming more important in many applications including self-driving cars, power grids, traffic networks, and integrated medical devices. Complex requirements for such real-life practical systems can be expressed as linear temporal logic formulae \cite{kesten1998algorithmic}. Model-based approaches for satisfying such requirements have been studied extensively in the literature \cite{girard2005reachability,BK08,tabuada09,belta2017formal}. In the setting of formal approaches for stochastic systems, a number of abstraction-based methods has been developed for the verification and synthesis of dynamical systems in order to either verify the desired specifications or synthesize controllers enforcing these systems to satisfy such specifications \cite{LAB15,majumdar2020symbolic,SVORENOVA2017230,zamani2014symbolic}.
%\Sadegh{you should add references to others as well, provide reference to Antuine Girard, Paolo Tabuada's book, and George Pappas's early paper.}
In order to improve  scalability of abstraction-based methods, some other techniques such as sequential gridding \cite{esmaeil2013adaptive,esmaeil2015faust}, discretization-free abstraction \cite{zamani2017towards}, and compositional abstraction-based techniques \cite{soudjani2015dynamic} have been introduced in the literature in order to efficiently deal with the verification and synthesis problems. 

\payam{An approach for formal verification and synthesis with respect to safety specifications in dynamical systems is to use a notion of barrier certificates \cite{prajna2004safety}}.
%in order to formally verify the safety of nonlinear and hybrid systems. 
%The reason behind choosing the name "barrier" is their application in optimization theory as they are added to the cost function in order to avoid unacceptable districts \cite{ames2019control}. 
%The results of this paper is in continue of Nagumo's theorem \cite{nagumo1942lage}.
%Let us consider a set $X$, an unsafe subset $X_{u}$, and an initial subset $X_{nit}$. The safe subset is the complement of the unsafe set $X_u^c$. If there exists a function $\mathrm{B}(x):\mathbb{R}^n\rightarrow \mathbb{R}$ where $\mathbb{B}(x)\leq0$ for all $x\in X_{init}$ and $\mathrm{B}(x)>0$ for all $x\in X_u$. Then $\mathrm{B}(x)$ is barrier certificate if $$\dot{\mathrm{B}}(x)\leq 0 \Rightarrow X_u^c\quad \text{is invarient}$$ 
Barrier certificates have been the focus of the recent literature as an abstraction-free technique that is scalable with the dimension of the system, i.e., they do not require construction of an abstraction of the system and \payam{can} provide directly the controller together with the guarantee on \payam{the} satisfaction of the safety specification \cite{zhang2010safety}, \cite{yang2020efficient}, \cite{borrmann2015control}. A barrier-based methodology is introduced in \cite{prajna2004safety} in order to verify safety in deterministic hybrid systems. In \cite{prajna2007framework}, a framework is proposed for safety verification of stochastic systems using barrier certificates which is extended to stochastic hybrid systems. The authors in \cite{wang2017safety} present barrier certificates that ensure collision-free behaviors in multi-robot systems by minimizing  the difference between the actual and the nominal controllers subject to safety constraints. In \cite{sloth2012compositional}, a compositional analysis is proposed for verifying the safety of an interconnection of subsystems using barrier certificates. The \payam{results in} \cite{jagtap2019formal} uses barrier certificates for \payam{the} synthesis of controllers against complex requirements expressed as co-safe linear temporal logic formulas.

The common requirement of the approaches mentioned above is the fact that they need a mathematical model of the system. However, a precise model of dynamical systems is either not available in many application scenarios or too complex to be of any use. Therefore, there is a need to develop approaches which are capable of verifying or synthesizing controllers against safety specifications only based on collected data from the system.

\textbf{Related Literature.}
Data-driven methods have gained significant attentions recently for formally verifying \payam{some} desired specifications.
A data-enabled predictive control is introduced in \cite{ coulson2020distributionally} that utilizes noisy data of the system and produces optimal control inputs \payam{ensuring} the satisfaction of desired chance constraints with high probability. A data-driven model predictive control scheme is proposed in \cite{berberich2020data} which only requires initially measured input-output trajectories together with an upper bound on the \payam{dimension} of the unknown system. In \cite{tabuada2020data}, a methodology is developed in order to make a single-input single-output system stable \payam{only} based on data. The stability problem of black-box linear switching systems \payam{with desired confidences} is investigated in \cite{kenanian2019data} based on collected data. This approach is extended in \cite{wang2019data} by providing a methodology for computing the invariant sets of \payam{discrete-time black-box systems}. A novel Bayes-adaptive planning algorithm for data-efficient verification of uncertain Markov decision processes is introduced in \cite{wijesuriya2019bayes}. A framework is proposed in \cite{sadraddini2018formal} to provide a formal guarantee on data-driven model identification and controller synthesis.
In \cite{salamati2020data}, a methodology is developed for providing a probabilistic confidence over the verification of signal temporal logic properties for partially unknown stochastic systems based on collected data. The authors in \cite{plambeck2022} propose a framework to learn a decision tree as a model for a black box continuous system. 

\Ali{The work in \cite{dawson2022safe} develops a method to synthesize robust feedback controllers with safety and stability guarantees. In \cite{robey2021learning}, a data-driven approach is proposed in order to synthesize controllers for deterministic hybrid systems using barrier certificates while providing a correctness guarantee on the obtained barrier certificate. A data-driven, model-based approach is developed in \cite{abate2020formal} to provide stability guarantees using Satisfiability Modulo Theories (SMT). The authors in \cite{niu2021safety} developed a data-driven technique to synthesize controllers for unknown deterministic systems. The framework developed in \cite{clark2021control} computes barrier certificates for complete- and incomplete-information systems affected by Gaussian process and measurement noises under unbounded inputs.}

An optimization-based approach is \payam{proposed} in \cite{robey2020learning} to learn a control barrier certificate through safe trajectories under suitable Lipschitz smoothness assumption on the dynamical system. A sub-linear algorithm is developed in \cite{han2015sublinear} for the barrier-based data-driven model validation of dynamical systems which computes the barrier function using a large dataset of trajectories. In \cite{jagtap20202020control}, a two-step procedure is proposed to synthesize a controller for an unknown nonlinear system, where the first step is to learn a Gaussian process as a replacement of the unknown dynamics, and the second step is to construct the control barrier function for the learned dynamics.

A data-driven optimization called \emph{scenario convex program} (SCP) is introduced in \cite{calafiore2006scenario} to solve robust convex optimizations. This approach replaces the infinite number of constraints in the robust optimization with a finite number of constrained by sampling the uncertain variables from their distributions. The approach relates the feasibility of the SCP to \payam{that of} the robust optimization while providing bounds on the probability of violating the constraints. \payam{The results in} \cite{kanamori2012worst}  studies the same approach and relates worst-case violation of the constraints to the probability of their violation. While \cite{calafiore2006scenario,kanamori2012worst} focus on feasibility, the authors in \cite{esfahani2014performance} establish a quantitative relation between the optimal value of the robust optimization and its associated SCP.

The results of \cite{esfahani2014performance} are employed in \cite{nejat2021} for data-driven verification of dynamical systems using some inequalities characterizing barrier certificates. Our results presented here differ from the ones in \cite{nejat2021} in three main directions.
First, our approach is developed for stochastic dynamical systems subject to random disturbances with unknown distributions, while the work in \cite{nejat2021} is restricted to deterministic systems. Second, our approach also tackles controller synthesis problems, while \cite{nejat2021} only deals with the verification ones. Last but not least, we study a class of non-convex optimization problems that makes our approach applicable to larger classes of systems, while the result in \cite{nejat2021} is restricted to only  convex problems.

\smallskip

\textbf{Contributions.} \payam{Here, we propose formal verification and synthesis procedures for unknown stochastic systems with respect to safety specifications based on collected data.} We first cast a barrier-based safety problem
as a robust convex program (RCP). \payam{Solving the obtained RCP is hard in general because the unknown model of the system appears in the constraints. To tackle this issue,} we resort to a scenario-driven approach by collecting samples from the system. Using the results \payam{in} \cite{esfahani2014performance}, we connect the optimal solution of the acquired scenario convex program (SCP) \payam{with that of} the original RCP. We provide a lower bound on the safety probability of the \Ali{unknown stochastic system} using a certain number of data which is related to the desired confidence. We extend this result to provide a new confidence bound for a class of non-convex barrier-based safety problems. We conclude the paper by three case studies to illustrate the applicability of our approach.
\textbf{Outline.} The structure of this paper is as follows. Section~\ref{sec:problemset} gives the system definition and the problem statement, and presents the safety verification of \payam{stochastic} systems using barrier certificates. In Section~\ref{sec:datadriven}, we introduce the scenario convex program for the barrier-based safety problem and we connect its optimizer to that of the original optimization. Our approach for the safety verification of the unknown stochastic system is presented in Section~\ref{sec:safety}. In Section~\ref{sec:synthesis}, we explain our data-driven synthesis approach which enforces the safety specification with a certain confidence. An extension of the verification problem for a class of non-convex safety problems is discussed in Section~\ref{sec:datadriven2}.
To illustrate the effectiveness of our approach, three case studies are presented in Section~\ref{sec:case_study}. Finally, Section~\ref{sec:ref} concludes the paper.
\section{Preliminaries and Problem Statement}
\label{sec:problemset}
\subsection{Notations and Preliminaries}
The set of positive integers, non-negative integers, real numbers, non-negative real numbers, and positive real numbers are denoted by $\mathbb{N} := \{1,2,3,\ldots\}$, $\mathbb{N}_0 := \{0,1,2,\ldots\}$, $\mathbb{R}$, $\mathbb{R}_0^+$, and $\mathbb{R}^+$, respectively. \Ali{We denote the indicator function of a set $\mathscr A \subseteq X$ by $\mathbbm{1}_\mathscr A:X\rightarrow \{0,1\}$, where $\mathbbm{1}_\mathscr A(x)$ is $1$ if $x\in \mathscr A$, and $0$ otherwise}. \payam{Notation} $\mathbf{1}_m$ is used to indicate a column vector of ones in $\mathbb{R}^{m\times1}$. We denote by $\Vert x\Vert$ the Euclidean norm of any $x\in\R^n$. We also denote the induced norm of any matrix $A\in\mathbb R^{m\times n}$  by $\|A\| = \sup_{x\neq 0} \|Ax\|/\|x\|$.
%The notation $||A||$ is used in order to denote the Euclidean norm of a matrix. If $\{v_1,\cdots,v_m\}$ is an orthonormal basis of eigen vectors with the corresponding eigen values $\{\lambda_1,\cdots,
%\lambda_m\}$ such that $|\lambda_1|\geq  ... |\lambda_m|$, then $||Av||^2=\lambda_1^2a_1^2+\cdots+\lambda_m^2a_m^2\leq \lambda_1^2||v||^2$ that results in $||A||\leq|\lambda_1|$. 
Given $N$ vectors $x_i \in \mathbb R^{n_i}$, $n_i\in \mathbb N$, and $i\in\{1,\ldots,N\}$, we use $[x_1;\ldots;x_N]$ and $[x_1,\ldots,x_N]$ to denote the corresponding column and row vectors, respectively, with dimension $\sum_i n_i$. The absolute value of a real number $x$ is denoted by $|x|$. \Ali{For a function $f:X\rightarrow Y$, we denote its inverse by $f^{-1}:Y\rightarrow X$, whenever exists. A regularized incomplete beta function for parameters $(z;a,b)$ is defined as $\mathrm I(z;a,b)=\frac{\int_{0}^{z}u^{a-1}(1-u)^{b-1}du}{\int_{0}^{1}u^{a-1}(1-u)^{b-1}du}$.}
If a system, \payam{denoted} by $\mathcal{S}$, satisfies a property $\Psi$ during a time horizon $\mathcal{H}$, \payam{it is denoted by $\mathcal{S} \models_{\mathcal{H}}\Psi$}. We also use $\models$ in this paper to show the feasibility of a solution for an optimization problem.

The sample space of random variables is denoted by $\Omega$. The Borel $\sigma$-algebras on a set $X$ is denoted by $\mathfrak{B}(X)$. The measurable space on $X$ is denoted by $(X,\mathfrak{B}(X))$. We have two probability spaces in this work. The first one is represented by $(X,\mathfrak{B}(X),\pr)$ which is the probability space defined over the state set $X$ with $\pr$ as a probability measure. The second one, $(V_w,\mathfrak{B}(V_w),\pr_w)$, defines the probability space over $V_w$ for the random variable $w$ affecting the stochastic system with $\pr_w$ as its probability measure.
With \payam{a slight abuse} of \payam{the} notation, we use the same $\pr$ and $\pr_w$ when the product measures are needed in the formulations.
Considering a random variable $z$, $\text{Var}(z):=\EE(z^2) - (\EE(z))^2$ denotes its variance with $\EE$ being the expectation operator.
%\Ali{
%By $\mathrm{B}(b,x), x\in X$, we denote a barrier certificate whose coefficients vector is $b$, and $x$ is a point in the state set $X$.
%}

%Helly's dimension of the optimization problem $\mathtt{Opt}[.]$ is defined as the least upper bound number of the support constraints of that problem for finitely many samples $N$ and can be defined formally as:
%$\textit{inf}~\{\zeta:\mu(|\texttt{Sc}(\mathtt{Opt}[.])|)\geq\zeta\}=0$
%where $\
%|\texttt{Sc}(.)|$ denotes the number of support constraints, and $\mu(.)$ is a probability measure.
\subsection{System Definition}
In this work, we first deal with (potentially) unknown discrete-time continuous-space stochastic \payam{dynamical} systems as formalized next.

\begin{definition}
	\label{def:mainsys}
	A discrete-time stochastic system (dt-SS) is a tuple $\mathcal{S} = (X,V_w,w,f)$, where
	the Borel set $X\subset \mathbb{R}^n$ is the state set of the system,
	the Borel set $V_w$ is the uncertainty space, 
	$w:=\{w(t):\Omega \rightarrow V_w, t\in\mathbb{N}_0\}$ is a sequence of independent and identically distributed (i.i.d.) random variables on the Borel space $V_w$  with some distribution $\pr_w$,
	and the map $f:X\times V_w\rightarrow X$ is a measurable function that characterizes the state evolution of the system.
	\payam{The state trajectory of the system is constructed according to}
	\begin{equation} 
		\mathcal{S}: x(t+1)=f(x(t),w(t)), \quad t\in\mathbb{N}_{0}.
		\label{eq:mainsystem}
	\end{equation}
	%where $X$ and $V_w$ are Borel $\sigma$-algebra on the state space $\mathbb{R}^n$ and uncertainty spaces, respectively. Variable $x$ denotes the state of the system as $x:=\{x(t):\Omega \rightarrow X, t\in\mathbb{N}_0\}$. 
	%
	We denote a finite trajectory of the system by $\xi(t):=x(0)x(1)\ldots x(t)$, $t\in \mathbb{N}_0$.
\end{definition}

In this work, we assume that the map $f$ and the distribution of the uncertainty $\pr_w$ are unknown. Instead, we assume we can \payam{collect} $N$ \emph{independent} and \emph{identically distributed} \payam{state pairs $(x_i,x^+_i)$} by initializing the system at $x_i$ and observing its next state as \Ali{$x^+_i=f(x_i,w_i)$ for some random sample $w_i$}. The collected \Ali{dataset is denoted by}
\begin{equation}
	\mathcal{D}:=\Big\{(x_i,x^+_i)\Big\}\subset X^2,\quad i \in\{1,\cdots,N\}. 
	\label{eq:data}
\end{equation}

% As the system is considered to be unknown, we collect data from the system and we provide a probabilistic confidence over the safety of the unknown stochastic system only based on the collected data. 

%  In this paper, we provide two different approaches in order to solve the problem stated in Subsection~\ref{sub:ps}. In the first one, we try to solve a Chance Constraint Problem ($CCP_{\epsilon}$) and we do not need any assumptions on the characteristic function $f(.,.)$ but an upper bound for its dimension. In the latter one, some nonrestrictive assumptions are added in order to solve a Robust Convex Programming and trying to quantify the difference between the objective function of the original $RCP$ problem and the obtained Scenario Convex Programming from the sampled data.
\subsection{Problem Statement}
\label{sub:ps}

%Next definition introduces the safety specification for the unknown stochastic system in Definition~\ref{def:mainsys}.
%\begin{comment}
\begin{definition}\label{safe}
	Given a set of initial states $X_{in}\subset X$, a set of unsafe states $X_u\subset X$, and a finite time horizon $\mathcal{H} \in \mathbb{N}_0$, the system $\mathcal{S}$ is called safe if all trajectories of $\mathcal{S}$ that start from $X_{in}$ never reach $X_u$ within horizon $\mathcal{H}$. We denote this safety property by $\Psi$ and its satisfaction by $\mathcal{S}$ is written as $\mathcal{S}\models_{\mathcal{H}}  \Psi$. \Ali{A state set $X$ containing the initial and unsafe sets is illustrated in Fig.~\ref{fig:barrier}}.
	%\Sadegh{replace $\models_{\mathcal{H}}$ with $\models$. You are not changing $\mathcal H$.}
\end{definition}

\begin{figure}
	\centering 
	\includegraphics[width=0.5\linewidth]{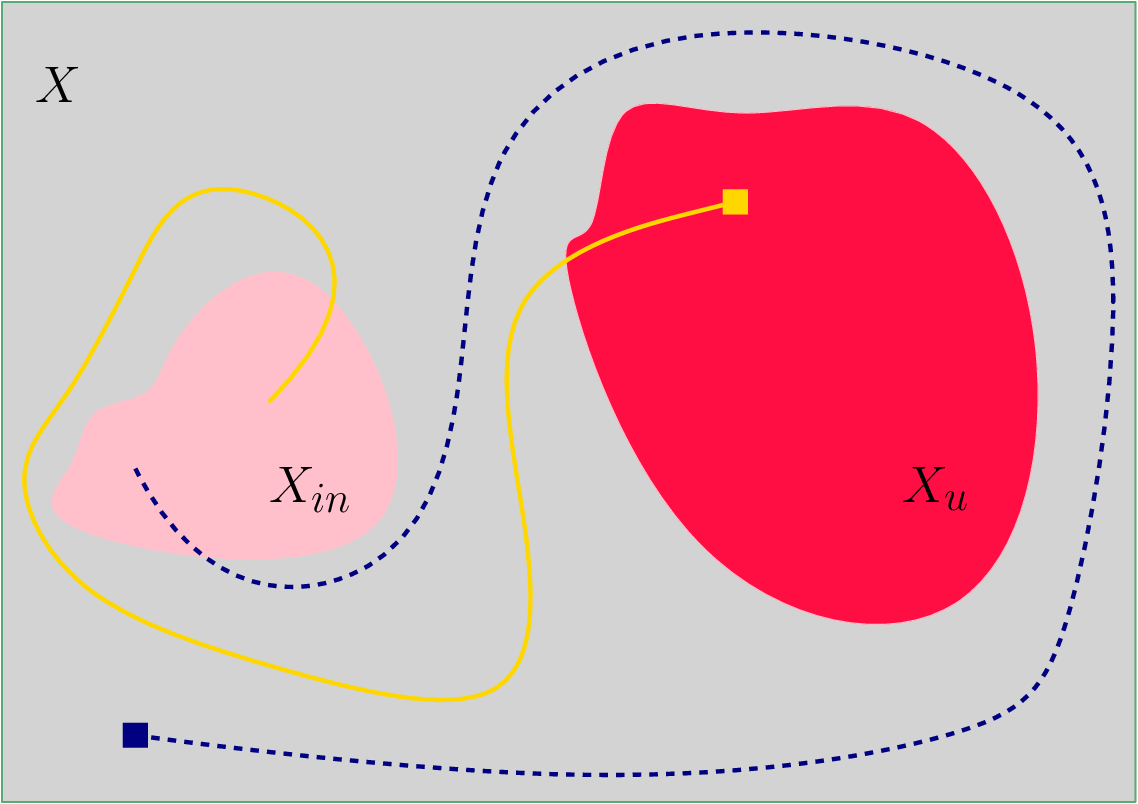}
	\caption{A set $X$ containing initial and unsafe sets $X_{in}$ and $X_u$. The blue dashed line illustrates a safe trajectory of the system, whereas the yellow one demonstrates an unsafe trajectory.
		%\payam{It seems that this boldness comes from the style of the template for mathematical expressions. We can either keep it, or we can use X\textsubscript{in} or X\textsubscript{u} instead. If we want to go for the second choice, I have to change them throughout the paper.}
		%\Sadegh{The TAC template  does not have this feature. There is an error somewhere that creates this issue. I leave it for now now.}
	}
	\label{fig:barrier}
\end{figure}

%\end{comment}
Since the system is stochastic and we do not know the distribution of $w$ and the map $f$,
%the question of interest here is: ``\emph{can one judge about the safety of a stochastic system only by leveraging data collected from trajectories of the system?}" This inspiring question can be formally represented as the following problem.
we are interested in establishing a lower bound on the probability that the safety property $\Psi$ is satisfied by the trajectories of $\mathcal S$ while using only a dataset of the form \eqref{eq:data}.
%\vspace{0.2cm}
%\begin{mdframed}
\payam{Now, we state the main problem we are interested to solve here.}
\begin{problem}
	\label{prob:problem}
	Consider an unknown dt-SS $\mathcal{S}$ as in Definition~\ref{def:mainsys}. Provide a lower bound $(1-\rho) \in [0,1]$ on the probability of satisfying $\Psi$, \emph{i.e.,}
	\begin{equation*}
		\pr_w\big(\mathcal{S}\models_{\mathcal{H}} \Psi\big)\ge 1-\rho,
	\end{equation*}
	together with a confidence $(1-\beta) \in [0,1]$ using only a dataset $\mathcal{D}$ of the form \eqref{eq:data}. Moreover, establish a connection between the required size of dataset $\mathcal{D}$ and the desired confidence $1-\beta$.
\end{problem}

%\end{mdframed}
%\vspace{0.2cm}

Therefore, we are interested in finding a \payam{potentially} tight lower bound. \Ali{The confidence $1-\beta$ in the statement of the problem is with respect to the probability distribution of the dataset $\mathcal{D}$ and is seen from the frequentist interpretation of probability: any algorithm that solves this problem collects dataset $\mathcal{D}$ using a probability distribution; while running the algorithm multiple times with different datasets $\mathcal{D}$, the algorithm gives wrong results (incorrect lower bound on the safety probability) in at most $\beta$ portion of the algorithm runs.} 

Fig.~\ref{fig:structure} shows an overview of our approach. \Ali{The block on the left represents a stochastic safety problem. The RCP block reformulates the safety problem as a robust optimization problem. Blocks SCP\textsubscript{$\scriptscriptstyle N$} and SCP\textsubscript{$\scriptscriptstyle N,\scriptscriptstyle \hat{N}$} solve the optimization problem introduced by the RCP block using finite number of samples. Finally, Theorem~\ref{theo:peyman} connects SCP's solutions to the original safety problem.}

%In order to address this problem, we first present the safety analysis of stochastic systems via barrier certificates in the next section.

\begin{figure}[ht]
	\centering 
	\includegraphics[width=.7\linewidth]{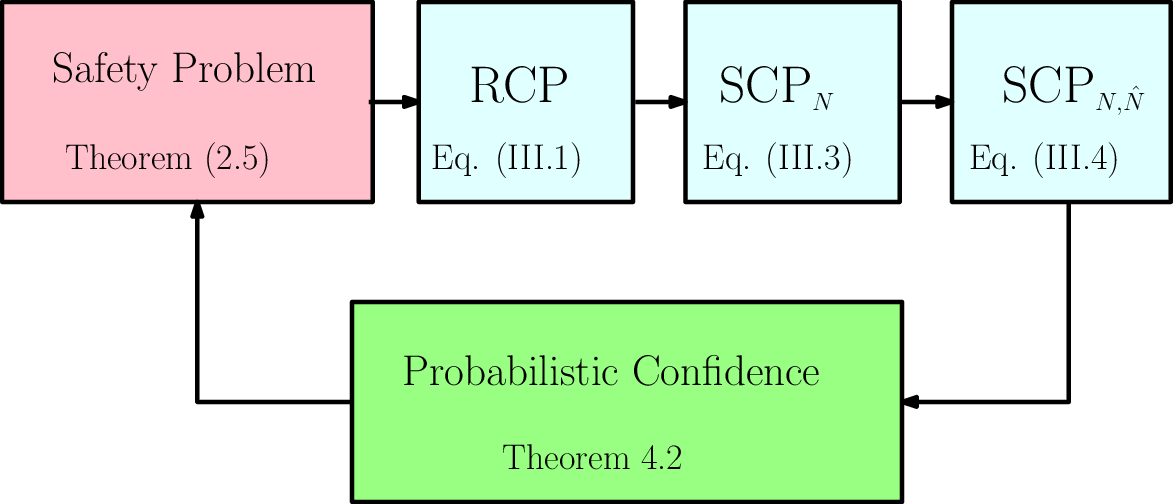}
	\caption{This figure shows an overview of the proposed scenario approach for verification of the safety specification.
		%\payam{It was not mentioned to use it in adjective form. I had changed "Schematic of" to "\{A\} schematic of" beforehand. I also checked that using search engines, and both forms have been used in famous references like the New York Times and other scientific papers with the same rate. But after your comment, I checked that on Longman  dictionary and realized that the adjective form is correct as you mentioned. Now, I corrected it.}
	}
	\label{fig:structure}
\end{figure}
\subsection{Safety Verification via Barrier Certificates}
\label{sec:barrier}
%Now we define a notion of probabilistic safety which will be used throughout the paper.
%\begin{definition}
%	Let us consider the system in \eqref{eq:mainsystem}. This system is considered to be \emph{stochastically safe} if and only if the trajectories of it which initialized at an initial subset $X_{nit}$ be safe trajectories, i.e., they \emph{always} remain in a with a probability greater than a threshold.
%	\label{def:safety}
%	\end{definition}
% The above inequality means that the trajectories of the stochastic system in \eqref{eq:mainsystem} never touch subset $X_u$ with a probability greater than one defined threshold value.  \AL{Abolfazl: again, this is redundant! This is the statement of the definition, why you rephrase it again?}

%In this section, we introduce the notion of barrier certificate (BC) that gives a lower bound on the probability of satisfying the safety specification.
%Let us start by formally defining the barrier certificates.
\begin{definition}
	Given a dt-SS $\mathcal{S} = (X,V_w,w,f)$, a nonnegative function $\mathrm{B}:X\rightarrow\mathbb{R}_{0}^{+}$ is called a barrier certificate (BC) for $\mathcal{S}$ if there exist constants $\lambda > 1$ and $c\in \mathbb{R}_{0}^{+}$ such that
	\begin{align}
		&\mathrm{B}(x)\leq 1,\qquad\qquad\qquad\qquad\quad~~~\forall x\in X_{in},	\label{eq:bar1}\\
		&\mathrm{B}(x)\geq \lambda, \qquad\qquad\qquad\qquad\qquad\;\!\forall x\in X_u,\label{eq:bar2}\\
		&\EE\Big[\mathrm{B}(f(x,w))\mid x\Big]\leq~\mathrm{B}(x)+c,\quad\quad\forall x \in X,\label{eq:bar3}
	\end{align}
	where $X_{in}\subset X$ and $X_u\subset X$ are initial and unsafe sets corresponding to a given safety specification $\Psi$, respectively. 
	\label{def:barrier}
\end{definition}
\vspace{.2cm}
%\begin{remark}
%	Throughout the main part of this paper, $\kappa$ is considered equal to $1$ which enables us to construct a robust convex problem to deal with. In section~\ref{sec:datadriven2}, an extension of the safety verification methodology for $\kappa \in [0,1]$ is presented which deals with a bilinear program.  
%\end{remark}

%\begin{remark}
%	Note that we enforce $\gamma$ to be less than $1$ (\emph{i.e.,} $\gamma< 1$) in Definition~\ref{def:barrier} in order to provide a meaningful probabilistic bound in Theorem~\ref{theo:kushner}.
%\end{remark}

Next theorem, borrowed from \cite{jagtap2019formal}, provides a lower bound on the probability of satisfaction of the safety specification for a dt-SS.

\begin{theorem}
	\label{theo:kushner}
	Consider a dt-SS $\mathcal{S}$ and a safety specification $\Psi$. Assume there exists a non-negative barrier certificate $\mathrm{B}(x)$ which satisfies conditions \eqref{eq:bar1}-\eqref{eq:bar3} with constants $\lambda$ and $c$. Then
	\begin{align}
		\pr_w\big(\mathcal{S}\models_{\mathcal{H}}  \Psi\big)\ge 1- \frac{1+c\;\mathcal{H}}{\lambda},
		\label{safety_kushi}
	\end{align}
	%	where,
	%	%	in which if $\kappa \in (0,1)$ 
	%	%	\begin{equation}
	%	%	\Delta = \begin{cases}
	%	%	1-(1-\frac{\gamma}{\lambda})(1-\frac{c}{\lambda}) \qquad\qquad\qquad\qquad\text{if $\lambda \ge \frac{c}{\kappa}$} \\
	%	%	\frac{\gamma}{\lambda}(1-\kappa)^{T_d}+\frac{c}{\kappa \lambda}\big(1-(1-\kappa)^{T_d}\big)  \;\qquad\text{if $\lambda < \frac{c}{\kappa}$},
	%	%	\end{cases}
	%	%	\label{eq:coschner1}
	%	%	\end{equation}
	%	\begin{align*}
	%	\Delta= \frac{1+c\;\mathcal{H}}{\lambda},
	%	\end{align*}
	with $\mathcal{H} \in \mathbb{N}_0$ being the finite time horizon \payam{associated with $\Psi$}.
	\label{theo:barrier1}
\end{theorem}

%\begin{remark}
%The reachability probability in \eqref{eq:coschner2} is less conservative with the constant $\Delta$ in \eqref{eq:coschner1} than with the one in \eqref{eq:coschner2} in terms of $\kappa$ as we might not find a feasible solution for $\kappa<1$. BUT in terms of probabilistic stuff, the first one gives a better probability. Then if we can solve the first one, we should stick to that. 
%\end{remark}
%\begin{corollary}
%	\label{coro:originalbarrier}
%	According to the Theorem~\ref{theo:barrier1}, a \emph{lower bound} can be considered for the safety of the system in \eqref{eq:mainsystem} in the sense of Definition~\ref{def:safety} as
%	\begin{equation}
%	\label{eq:safetycondition}
%	\pr(\mathcal{S}\models \mathtt{Safety})\ge 1-\Delta.
%	\end{equation}
%	\end{corollary}
\vspace{0mm}

In this work, we consider polynomial-type barrier certificates denoted by $\mathrm{B}(b,x)$, where $b$ is the vector containing the coefficients of the polynomial. \payam{Such a polynomial} with degree $k\in \mathbb{N}_0$ has the form
\begin{align}
	\mathrm{B}(b,x)=\sum_{\iota_1=0}^k\ldots\sum_{\iota_n=0}^kb_{\iota_1,\ldots,\iota_n}(x_1^{\iota_1}\ldots x_n^{\iota_n}),
	\label{eq:fixedbarrier}
\end{align} 
with $b_{\iota_1,\ldots,\iota_n}=0$ for $\iota_1+\ldots+\iota_n>k$. Hence, finding a polynomial barrier certificate reduces to determining the coefficients of the polynomial, namely $b_{\iota_1,\ldots,\iota_n}$. In the next section, we provide our data-driven approach for the construction of \payam{polynomial-type} barrier certificates.

%\Sadegh{You had a mistake here. The number of variables $x_1,\ldots,x_n$ is different from the degree of the polynomial, which is $k$. I corrected the  form of the polynomial. Make sure to correct the rest of the paper. You should replace $b_{\iota_1,\ldots,\iota_k}$ with $b_{\iota_1,\ldots,\iota_n}$.}
\section{Data-driven Safety Verification}
\label{sec:datadriven}
We first cast the barrier-based safety problem in Theorem~\ref{theo:kushner} as a robust convex programming (RCP). We then provide a scenario-based approach in order to solve the obtained RCP using data collected from the system.

Satisfying the conditions of Theorem~\ref{theo:kushner} is equivalent to having a non-positive value for the optimal solution of the following RCP (i.e., $\mathcal{K}\le 0$):
\begin{align} 
	\text{RCP}:\left\{
	\begin{array}{ll}
		\underset{d}{\min}\quad \mathcal{K}\quad \quad\\
		\text{s.t.}
		\quad \max_z\big(g_z(x,d)\big) \!\leq\! 0, z \!\in\!\{1,\dots,5\},\forall x \!\in\! X,\\
		\quad ~~ \;\;d=[\mathcal{K};\lambda;c;b_{\iota_1,\ldots,\iota_n}],\\
		\quad ~~ \;\;\mathcal{K}\in \mathbb{R},\; \lambda > 1,\; c \geq 0, 
	\end{array}
	\right.
	\label{eq:RCP}
\end{align}
in which,
\begin{align}
	&g_1(x,d)=-\mathrm{B}(b,x)-\mathcal{K},\nonumber\\
	&\Ali{g_2(x,d)=(\mathrm{B}(b,x)-1-\mathcal{K})\mathbbm{1}_{X_{in}}(x),} \nonumber\\
	&\Ali{g_3(x,d)=(-\mathrm{B}(b,x)+\lambda-\mathcal{K})\mathbbm{1}_{X_{u}}(x),}\nonumber\\
	&g_4(x,d)=\frac{1+c\;\mathcal{H}}{\rho}-
	\lambda-\mathcal{K},\nonumber\\
	&g_5(x,d)=\EE\Big[\mathrm{B}(b,f(x,w))\mid x\Big]-\mathrm{B}(b,x)-c-\mathcal{K},
	\label{eq:sampledcondition}
\end{align}
where $(1-\rho)$ is a given lower bound for the safety probability. 

\begin{remark}
	The RCP~\eqref{eq:RCP} is in fact a robust convex optimization. It is a convex optimization since the constraints are convex with respect to decision variables in $d$ and \payam{objective function.} It is a robust optimization since the constraints have to hold for all $x\in X$. 
\end{remark}

\begin{remark}
	The RCP~\eqref{eq:RCP} always has a feasible solution. For instance, by choosing coefficients of $\mathrm{B}(b,x)$ equal to zero, $\lambda=2$, $c=0$, and $\mathcal{K}\geq\frac{1}{\rho}-2$, we get a feasible solution for the RCP. Moreover, the barrier certificate obtained from this RCP satisfies conditions \eqref{eq:bar1}-\eqref{eq:bar3} \payam{as long as} $\mathcal{K}\le 0$.
\end{remark}

Finding an optimal solution for the RCP in \eqref{eq:RCP} is hard in general because the map $f$ is unknown, the probability measure $\pr_w$ is also unknown (thus the expectation in $g_5$ cannot be computed analytically), and there are infinitely many constraints in the robust optimization since $x\in X$, \payam{where $X$ is a continuous set.} 
%not only is there no access to the model of the system (\emph{i.e.,} $f$), but also the state of the system lives in an infinite set $X$.
To tackle \payam{this}, we \payam{first} assign a probability distribution to the \payam{state set}, take $N$ i.i.d.  samples  $\{x_1, x_2,\ldots, x_N\}$ from this distribution, and replace the robust quantifier $\forall x\in X$ with $\forall  x_i\in X$, $i\in\{1,2,\ldots,N\}$. This results in the following scenario convex program denoted by SCP\textsubscript{$\scriptscriptstyle N$}:
%\Sadegh{We need to add a remark in this section on the choice of $\pr$ that we put on the state space. Our result is independent on this choice. We select uniform distribution.}
%this problem, we collect data from the trajectories of the unknown system and propose a corresponding scenario convex program of RCP, denoted by SCP\textsubscript{N} as following:
%\begin{align}
%\underset{\gamma,c}{\mathtt{min}}\quad\quad&\gamma+cT_d\quad \quad \nonumber\\\mathtt{s.t.}\quad
%&\max\big\{g_1(\hat{x}_i), g_2(\hat{x}_i, \gamma), g_3(\hat{x}_i), \bar{g}_4(\hat{x}_i,c)\big\} \leq 0,
%\label{eq:sampledoptima}
%\end{align}
%with 
%\begin{align}
%&\bar{g}_4(\hat{x}_i,c,\delta)=\mathbb{E}\Big[\mathrm{B}(f(\hat{x}_i,w))\mid \hat{x}_i\Big]-  \;\mathrm{B}(\hat{x}_i)-c+\delta,
%\label{eq:last}
%\end{align}
%where
%$	\mathbb{E}[\mathrm{B}(\hat{f}(\hat{x}_i,w_j))\mid \hat{x}_i]$ can be approximated by the empirical mean as $\frac{\sum_{j=1}^{\hat{N}}\mathrm{B}(f(\hat{x}_i,\hat{w}_j))}{\hat{N}}$for $\hat{N}$ number of samples.
\begin{align} 
	\text{SCP\textsubscript{$\scriptscriptstyle N$}}:\left\{
	\begin{array}{ll}
		\underset{d}{\min}\quad \mathcal{K}\quad \quad\\
		\text{s.t.}\quad \;
		\max_z\,g_z(x_i,d) \!\leq\! 0, \,\,\forall i\in \{1,\ldots,N\},\\
		\qquad\qquad ~\;\;z \!\in\!\{1,\dots,5\},\\
		\qquad ~~~d=[\mathcal{K};\lambda;c;b_{\iota_1,\ldots,\iota_n}],\\
		\qquad ~~~\mathcal{K}\in \mathbb{R},\; \lambda > 1,\; c \geq 0.
	\end{array}
	\right.
	\label{eq:SCN1}
\end{align}

To tackle the issue of unknown $\pr_w$, 
%Since there is no closed-form solution for the expected value in the last condition of \eqref{eq:sampledcondition}, 
we replace the expectation in $g_5$ with its empirical approximation by sampling $\hat N$ i.i.d. values $w_j,\; j \in \{1,\ldots,\hat{N}\}$, from $\pr_w$ for each $x_i$, which gives the following scenario convex program denoted by SCP\textsubscript{$\scriptscriptstyle N,\hat{N}$}:
\begin{align} 
	\text{SCP\textsubscript{$\scriptscriptstyle N,\hat{N}$}}\!:\!\left\{
	\begin{array}{ll}
		\underset{d}{\min}\quad \mathcal{K}\quad \quad\\
		\text{s.t.}\quad\;
		\max_z\,\bar g_z(x_i,d) \!\leq\! 0, \,\,\forall i\in \{1,\ldots,N\},\\
		\qquad\qquad ~\;\;z \!\in\!\{1,\dots,5\},\\
		%\max\big(g_z(x_i,d),\bar{g}_5(x_i,d)\big) \!\!\leq\! 0,z \!\in\!\{1,\ldots,4\},\quad \\
		%\quad\quad~~ \forall x_i \in X, \forall i\in\{1,\ldots,N\},\\
		\qquad ~~d=[\mathcal{K};\lambda;c;b_{\iota_1,\ldots,\iota_n}],\\
		\qquad ~~\mathcal{K}\in \mathbb{R},\; \lambda > 1,\; c \geq 0,
	\end{array}
	\right.
	\label{eq:SCN2}
\end{align}
where $\bar g_z:=g_z$ for all $z\in\{1,2,3,4\}$ and
\begin{align}
	\bar{g}_5(x_i,d):=\frac{1}{\hat{N}}\sum_{j=1}^{\hat{N}}\mathrm{B}(b,f(x_i,w_j))-  \;\mathrm{B}(b,x_i)-c+\delta-\mathcal{K}.
	\label{eq:empirical}
\end{align}
%\Sadegh{look at how I kept a consistent notation in the above optimisation. do the same thing for the controlled case.}
In SCP\textsubscript{$\scriptscriptstyle N,\hat{N}$}, $f(x_i,w_j)$ is the next state of the system from \payam{the current state} $x_i$ with the noise realization $w_j$. Therefore, the solution of the SCP\textsubscript{$\scriptscriptstyle N,\hat{N}$} can be obtained using only the dataset $\mathcal D$ without \payam{the} knowledge of $f$ and $\pr_w$.
The optimal value for the objective function of SCP\textsubscript{$\scriptscriptstyle N,\hat{N}$} is denoted by $\mathcal{K}^*(\mathcal{D})$. We also denote by $\hat{\mathrm{B}}(b,x \,|\, \mathcal{D})$ the barrier function constructed based on the solution of SCP\textsubscript{$\scriptscriptstyle N,\hat{N}$} in \eqref{eq:SCN2}.

Note that $\bar{g}_5(x_i,d)$ in \eqref{eq:empirical} has an additional parameter $\delta> 0$ compared to $g_5$. This \payam{parameter} is added to make the last inequality more conservative in order to capture the error coming from replacing the expectation with the empirical mean.
%in which $\hat{N}\in\mathbb{N}$ and $\delta\in \mathbb{R}^+$ are respectively the number of samples required for the empirical approximation, and the error introduced by the approximation.
%
%The empirical approximation for each sample $\hat{x}_i$ is computed over $\hat{N}$ different realizations of noise $\hat{w}_j,j\in\{1,\dots,\hat{N}\}$. This approximation introduces an error in the last condition, represented by $\delta$, which makes the last condition more conservative.
We use Chebyshev's inequality \cite{hernandez2001chebyshev} to quantify such an error with the associated confidence. Let us define the variance of the empirical approximation as
\begin{align}
	\sigma^2 := \text{Var}\Big(\frac{1}{\hat{N}}\sum_{j=1}^{\hat{N}}\mathrm{B}(b,f(x,w_j))\Big),
	\label{eq:variance_1}
\end{align}
\payam{where the variance is taken with respect to $w_j$.} We assume that there is a bound $\hat M$ such that
\begin{equation}
	\label{eq:variance_2}
	\text{Var}\big(\mathrm{B}(b,f(x,w)\big)\leq \hat{M},\quad \forall x\in X.
\end{equation}
This assumption gives us a bound for $\sigma^2$ in \eqref{eq:variance_1} as $\sigma^2\le \frac{\hat M}{\hat N}$ due to $w_j$ being independent.
The idea of replacing the expectation by the empirical mean in an optimization problem and relating the associated solutions based on Chebyshev's inequality is also used in \cite{SM18_Concentration}.
Next theorem shows that the barrier certificate computed using the optimal solution of the SCP\textsubscript{$\scriptscriptstyle N,\hat{N}$} is a feasible barrier certificate for SCP\textsubscript{$\scriptscriptstyle N$} in \eqref{eq:SCN1} with a certain confidence.

\begin{theorem}
	\label{theo:empirical}
	Let $\hat{\mathrm{B}}(b,x\,|\,\mathcal{D})$ be a feasible solution of the SCP\textsubscript{$\scriptscriptstyle N,\hat{N}$} for some $\delta> 0$, and \payam{assume} the inequality~\eqref{eq:variance_2} holds with a given $\hat{M}$. \payam{Then for any $\beta_s \in (0,1]$, we get}
	\begin{equation}
		\label{eq:feasibility123}
		\pr_w\Big(\hat{\mathrm{B}}(b,x\,|\,\mathcal{D}) \models\text{SCP\textsubscript{$\scriptscriptstyle N$}}\Big)\geq 1- \beta_s,
	\end{equation}
	provided that the number of samples in the empirical mean satisfies $\hat{N} \geq \frac{\hat{M}}{\delta^2\beta_s}$.
	% a desired stochastic confidence $\beta_s \in (0,1]$, and a given upper bound $\hat{M}$ on the variance of the barrier function applied on $f$, \emph{i.e.,} $\text{Var}\big(\mathrm{B}(b,f(x,w)\big)\leq \hat{M}\in \mathbb{R}^+$, one has
\end{theorem}
\vspace{.2cm}
\begin{proof}
	\payam{By the statement of the theorem, we have}  $\hat{\mathrm{B}}(b,x\mid \mathcal{D})\models\text{SCP\textsubscript{$\scriptscriptstyle N,\hat{N}$}}$.
	%We know that the constructed barrier function with coefficients obtained from solving the SCP\textsubscript{w} in \eqref{eq:SCN2} is a feasible solution for that problem namely $$\hat{\mathrm{B}}(b,x\mid \mathcal{D})\models\text{SCP\textsubscript{w}}.$$
	The difference between the empirical mean in \eqref{eq:empirical} and the expected value in \eqref{eq:SCN1} can be quantified by invoking the Chebyshev's inequality \payam{as:}
	\begin{equation}
		\pr_w\Big(\abs{\EE\big [\mathrm{B}(b,f(x,w))\!\mid\! x\big ]-\frac{1}{\hat{N}}{\sum_{j=1}^{\hat{N}}\!\mathrm{B}(b,f(x,w_j))}}\!\leq\!\delta\Big)\!\geq\! 1\!-\!\frac{\sigma^2}{\delta^2},
		\label{eq:chebyschevproof}
	\end{equation}
	where $\delta\in \mathbb{R}^+$, and $\sigma^2$ is defined in \eqref{eq:variance_1} \cite{hernandez2001chebyshev}.
	%where $\delta$ and $\beta_s$ are values regarding the introduced error and stochastic confidence, respectively.
	%Again, according to Chebyshev's theorem, one has $\beta_s=\frac{\sigma^2}{\delta^2}$.
	%The value of $\sigma^2$ is dependent to the number of samples $\hat{N}$ which is used in computing the empirical mean.
	Since all the first four feasibility conditions are the same \payam{as} in \eqref{eq:SCN1} and \eqref{eq:SCN2}, $\hat{\mathrm{B}}(b,x\,|\,\mathcal{D})$ is a feasible solution for those conditions of SCP\textsubscript{$\scriptscriptstyle N$} with probability one. The only remaining concern is the last feasibility condition. According to \eqref{eq:chebyschevproof}, one can deduce that $\hat{\mathrm{B}}(b,x\,|\,\mathcal{D})$ is a feasible solution for SCP\textsubscript{$\scriptscriptstyle N$} with a confidence \payam{of} at least $1-\frac{\sigma^2}{\delta^2}$. Furthermore, we have $\sigma^2\leq\frac{\hat{M}}{\hat N}$ by \payam{having} $\text{Var}(\mathrm{B}(b,f(x,w)))\leq\hat{M}$, and hence
	$$\pr_w\big(\hat{\mathrm{B}}(b,x\,|\,\mathcal{D})\models\text{SCP\textsubscript{$\scriptscriptstyle N$}}\big)\geq\! 1\!-\!\frac{\hat{M}}{\delta^2\hat{N}}.$$
	By the above inequality, we get $\beta_s\ge \frac{\hat{M}}{\delta^2\hat{N}}$ and consequently $\hat{N}\geq\frac{\hat{M}}{\delta^2\beta_s}$. This completes the proof.
\end{proof}

\medskip

\begin{remark}
	When the system has additive noise, \emph{i.e.,} 
	\begin{align}
		x(t+1) = f_a(x(t)) + w(t),\nonumber
	\end{align}
	the condition \eqref{eq:variance_2} can be established by having a bound on $f_a(\cdot)$ and bounds on moments of the noise $w$.
	%In this case, the value of $\hat{M}$ is computable using a bound on $f_a(x(t))$ and bounds on moments of $w$.
	%\AL{How $\hat{M}$ can be computed in that case? Not clear to me. Explain more.}
	%
	For instance, in the case of one-dimensional systems (\payam{i.e.,} $n=1$), we have $\mathrm{B}(b,x) = \sum_{\iota=0}^{k} b_{\iota} x^{\iota}$ and the variance of $\mathrm{B}(\cdot)$ can be expanded as follows:
	\begin{align*}
		& \text{Var}(\mathrm{B}(b,f(x,w))) = \text{Var}\Big(\sum_{\iota=0}^{k} b_{\iota} f(x,w)^{\iota}\Big)\\
		&=\!\text{Var}\Big(\!\sum_{\iota=0}^{k} b_{\iota} (f_a(x) + w)^{\iota}\Big) \!=\!\text{Var}\Big(\!\sum_{\iota}^{k}\!\sum_{j=0}^{\iota}\! b_{\iota} {\iota \choose j} \!f_a(x)^{\iota-j} w^{j}\!\Big)\\
		& =\!\text{Var}\Big(\sum_{j=1}^{k}\mathrm{g}_{j}(x) w^{j}\Big) \text{ with } \mathrm{g}_{j}(x) := \sum_{\iota=j}^{k} b_{\iota} {\iota \choose j} f_a(x)^{\iota-j}\\
		& =\! \sum_{j=1}^{k}\sum_{z=1}^{k} \mathrm{g}_j(x)\mathrm{g}_z(x)(\EE[w^{j+z}] - \EE[w^j]\EE[w^z]).
	\end{align*}
	This means the variance can be bounded using upper bounds of $f_a(\cdot)$ and moments of $w$.
\end{remark}

As it can be seen from Theorem~\ref{theo:empirical}, higher number of samples  $\hat{N}$ is needed in order to have a smaller empirical approximation error $\delta$, and to provide a better confidence bound. In fact, $\hat{N}$ and $\delta$ are required to solve the SCP\textsubscript{$\scriptscriptstyle N,\hat{N}$} in \eqref{eq:SCN2}. Later in the next section, we show how the value of  $\beta_s$ affects the total confidence concerning the safety of the stochastic system. 
\begin{remark}
	%	It should be noticed that we do not need to know the probability distribution $\pr$ over the state space. Instead, we collect $N$ i.i.d. samples from an arbitrary distribution; for instance we collect data from a uniform distribution.
	%\textcolor{red}{see how I rewrote this remark. Eliminate the color after  comparing with your version.}
	Note that our results presented in this paper are valid for any choice of the probability distribution $\pr$ \Ali{with its support being the state set $X$ that satisfies a regularity assumption formulated in the next section (cf.~Assumption~\ref{ass:G}).
		%for which a function $G$ introduced in the next section, can be computed such that satisfies \eqref{eq:gofr}.
		This assumption holds for a wide range of distributions including uniform, truncated normal, and exponential distributions.} From the algorithmic perspective, this distribution affects the collected data points $x_i$ and the optimal solution of the SCP\textsubscript{$\scriptscriptstyle N$}. The confidence formulated in our paper is also with respect to this distribution. %\Ali{It is clear that sampling distribution support needs to cover the whole state set $X$.}
	We choose $\pr$ to be a uniform distribution in the case study section.
\end{remark}
\section{Safety Guarantee over Unknown Stochastic Systems}
\label{sec:safety}
\label{sub:main_safe}
In the previous section, we established the connection between the two optimizations SCP\textsubscript{$\scriptscriptstyle N$} and SCP\textsubscript{$\scriptscriptstyle N,\hat{N}$}, and showed that the solution of SCP\textsubscript{$\scriptscriptstyle N,\hat{N}$} is a feasible solution for SCP\textsubscript{$\scriptscriptstyle N$} with a certain confidence if the number of samples $\hat N$ is chosen appropriately (cf. Theorem~\ref{theo:empirical}). In this section, we focus on the relation between the original RCP and the SCP\textsubscript{$\scriptscriptstyle N$} utilizing the fundamental result of \cite{esfahani2014performance} and provide an end-to-end safety guarantee over the unknown stochastic system with a priori guaranteed confidence.
%
%we showed that using a finite number of data, the original RCP can be corresponded to a scenario convex programs (SCP\textsubscript{N}) for which the solution can be approximated with any arbitrary precision (cf. Theorem~\ref{theo:empirical}). In this section, we establish the missing connection between solutions to the original RCP and the corresponding SCP\textsubscript{N} by employing the proposed fundamental results in \cite{esfahani2014performance}, and consequently, we provide a safety guarantee over the unknown stochastic system with a priori guaranteed confidence.
%In this section, we aim at providing a probabilistic confidence for the safety of potentially unknown stochastic systems based on a finite set of data collected from the system as in Problem~\ref{prob:problem}. In order to address this problem, we employ the proposed fundamental results in \citep{esfahani2014performance} and relate the optimizer of the SCP\textsubscript{w} to that of the RCP, and consequently, we provide a safety guarantee over the unknown stochastic system with a priori guaranteed confidence.
\Ali{To do so, we need to raise the following regularity assumptions on the functions and the chosen probability measure $\pr$}.
\begin{assumption}
	\label{assum:lip}
	Functions $g_1$, $g_2$, $g_3$, and $g_5$ are all Lipschitz continuous with respect to $x$ with Lipschitz constants $\mathrm{L}_{x_1}$, $\mathrm{L}_{x_2}$, $\mathrm{L}_{x_3}$, and $\mathrm{L}_{x_5}$, respectively.
	\Ali{
		Therefore, the Lipschitz constant $\mathrm{L}_{x} := \mathrm{L}_{x_1}+\mathrm{L}_{x_2}+\mathrm{L}_{x_3}+\mathrm{L}_{x_5}$ is a Lipschitz constant for $\max_z g_z(x,d), z\in\{1,\ldots,5\}\setminus\{4\}$.
		In addition, if $g_1$, $g_2$, $g_3$, and $g_5$ are analytic over a compact domain $X$, the Lipschitz constant of $\max_z g_z(x,d)$ is
		$\mathrm{L}_{x} := \max \big\{\mathrm{L}_{x_1},\mathrm{L}_{x_2},\mathrm{L}_{x_3},\mathrm{L}_{x_5}\big\}$.
	}
\end{assumption}

\Ali{
	\begin{assumption}
		\label{ass:G}
		There is a strictly increasing function $G:\mathbb{R}_0^+ \rightarrow [0,1]$, where $G(0)=0$ such that
		\begin{align}
			\pr[\mathrm{b}(x,r)]\geq G(r)\qquad \forall x \in X,
			\label{eq:gofr}
		\end{align}
		where $\mathrm{b}(x,r) \subset X$ is an open ball centered at point $x$ with radius $r$.
	\end{assumption}
}
\Ali{Note that any probability distribution, for which the above lower bound function $G(r)$ can be computed, can be used in our approach for sampling.
}
%We utilize Assumption~\ref{assum:lip} and propose the next theorem that establishes a relation between the optimal values of SCP\textsubscript{w} and that of the original RCP, and accordingly, verifies the safety of unknown stochastic systems with a priori guaranteed confidence.
\medskip
\begin{comment}
\begin{theorem}
\label{theo:peyman}
Consider an unknown dt-SS, as in~\eqref{eq:mainsystem}, and safety specification $\Psi$. 
Let Assumption~\ref{assum:lip} hold with constant $\mathrm{L}_{x}$. Assume $\hat{N}$ is selected for the SCP\textsubscript{$\scriptscriptstyle N,\hat{N}$} as in Theorem~\ref{theo:empirical} in order to provide confidence $1-\beta_s$. \Ali{If first we have $N\geq N\big(\epsilon,\beta\big)$, with
\begin{equation}
\label{eq:number_samples}
N(\epsilon,\beta) := \min \Big\{N\in\mathbb{N}\mid\sum_{i=0}^{\mathcal{\mathcal{Q}}+2} \dbinom{N}{i}\epsilon^{\;i}
(1-\epsilon)^{N-i} \leq \beta \Big\},
\end{equation}
where $\epsilon,\beta \in [0,1]$ together with $\mathcal{Q}$ as the number of coefficients of barrier certificate, and second we have 
\begin{align}
\mathcal{K}^\ast(\mathcal{D})+\mathrm{L}_{x}\;G^{-1}(\epsilon)\leq 0,
\label{eq:main_condition_1}
\end{align}
where $\mathcal{K}^\ast(\mathcal{D})$ is the optimal value of the optimization problem in \eqref{eq:SCN2} using $\hat{N}$ and $N$, and $G(\epsilon)$, as in \eqref{eq:gofr} for $r=\epsilon$,  statement holds with a confidence of at least $1-\beta-\beta_s$:\\
\begin{equation*}
\label{eq:lower_bound}
\pr_w\big(\mathcal{S}\models_{\mathcal{H}} \Psi \big)\ge1-\rho,
\end{equation*}
for some given $\rho\in(0,1]$.
}
\end{theorem}
\end{comment}
\Ali{
	\begin{remark}
		The probability distribution from which $x_i$ is sampled must satisfy Assumption~\ref{ass:G}.
		This assumption requires having a strictly increasing function $G:\mathbb{R}_0^+ \rightarrow [0,1]$ that satisfies
		\begin{align*}
			\pr[\mathrm{b}(x,r)]\geq G(r),\qquad \forall x \in X.
			%\label{eq:gofr}
		\end{align*}
		%Then, one can deduce that $G(0) = 0$ and $G(r)>0$ for all $r>0$.
		Then, the probability distribution $\pr$ should assign positive probability to any ball with positive radius. This means no ball $\mathrm{b}(x,r)\subset X$ could be excluded from sampling in the approach with some non-trivial probability. 
		% Since $\pr[\mathrm{b}(x,r)]\geq G(r), \forall x \in X,  \forall r > 0$ according to \eqref{eq:gofr}, one can easily deduce that $\pr[\mathrm{b}(x,r)]=0, \forall x \in X, \forall r > 0$. Then the following statement holds true.  
		% \begin{align*}
		% \nexists r>0, \forall x \in X,\;\; \pr[\mathrm{b}(x,r)]=0.
		% \end{align*}
		% This means that support of the sampling distribution must cover the whole sate set, i.e., $\text{supp}(\pr)=X$.
	\end{remark}
	Next, we introduce the main result which connects the safety of an unknown stochastic system directly to data collected from the system.
}
\begin{theorem}
	\label{theo:peyman}
	Consider an unknown dt-SS, as in~\eqref{eq:mainsystem}, and safety specification $\Psi$. 
	\Ali{Let Assumptions~\ref{assum:lip} and \ref{ass:G} hold with Lipschitz constant $\mathrm{L}_{x}$ and function $G(r)$, respectively.} Assume $\hat{N}$ is selected for the SCP\textsubscript{$\scriptscriptstyle N,\hat{N}$} as in Theorem~\ref{theo:empirical} in order to provide confidence $1-\beta_s$. \Ali{
		Denote by $\mathcal{K}^\ast(\mathcal{D})$ the optimal value of the optimization problem in \eqref{eq:SCN2} using $N$ samples and parameter $\rho\in(0,1]$.
		For any $\beta \in [0,1]$,
		the following statement holds with a confidence of at least $(1-3\beta-\beta_s)$:\\
		\begin{equation*}
			\label{eq:lower_bound}
			\pr_w\big(\mathcal{S}\models_{\mathcal{H}} \Psi \big)\ge1-\rho,
		\end{equation*}
		if 
		\begin{align}
			\mathcal{K}^\ast(\mathcal{D})+\mathrm{L}_{x}\;G^{-1}(\epsilon)\leq 0,
			\label{eq:main_condition_1}
		\end{align}
		where function $G$ defined in \eqref{eq:gofr}, and $\epsilon=\mathrm I^{-1}(1-\beta;\mathcal Q+3,N-\mathcal{Q}-2)$.
	}
\end{theorem}
\begin{proof}
	\Ali{
		Denote the optimal values of the RCP and the SCP\textsubscript{$\scriptscriptstyle N$} by $\mathcal K^\ast$ and $\mathcal K^\ast_{\mathsf m}(\mathcal D)$, respectively.
		%	Since the solutions to the RCP in \eqref{eq:RCP} and the safety problem in Theorem~\ref{theo:kushner} are equivalent, one has
		%	\begin{align}
		%	\pr_w\big(\mathcal{S}\models \Psi \big)\geq1-\Delta^*,
		%	\label{eq:first}
		%	\end{align}
		%Now, we want to use the results in \cite{esfahani2014performance} to relate the optimal solutions of \eqref{eq:SCN2} and \eqref{eq:RCP}.
		According to \cite[Theorem~3.6]{esfahani2014performance}, one has
		\begin{align*}
			\pr\big(\mathcal{K}^\ast_{\mathsf m}(\mathcal D)\leq \mathcal K^\ast\leq \mathcal{K}^\ast_{\mathsf m}(\mathcal D)+\mathrm{L}_{sp}H(\epsilon)\big)\geq1-\beta,
		\end{align*}
		for a chosen $\epsilon$ and any $N\geq N(\epsilon,\beta)$ as in \cite[Theorem 2.2]{esfahani2014performance}.
		Equivalently, the above inequality holds for a given $N$ and $\epsilon \leq \mathrm I^{-1}(1-\beta;\mathrm d,N-\mathrm d+1)$.
		In this expression, $\mathrm d$ is the number of decision variables, and $H(\cdot)$ is a uniform level-set bound as defied in \cite[Definition 3.1]{esfahani2014performance}.  Constant $\mathrm{L}_{sp}$ is a Slater constant as defined in \cite[equation (5)]{esfahani2014performance}.
		Since the original RCP in \eqref{eq:RCP} is a min-max optimization problem, the constant $\mathrm{L}_{sp}$ can be selected as one according to \cite[Remark~3.5]{esfahani2014performance}. By choosing $\mathrm{d}:=\mathcal{Q}+3$, one obtains the parameters of the incomplete beta function in the theorem statement. Based on \cite[Proposition 3.8]{esfahani2014performance}, $H(\epsilon) = \mathrm{L}_xG^{-1}(\epsilon)$, where $\mathrm{L}_x$ is the Lipschitz constant of RCP as in Assumption~\ref{assum:lip}, and $G(\cdot)$ as in \eqref{eq:gofr}.
		Now, one can readily deduce that
		\begin{equation}
			\pr\big(\mathcal{K}^*\leq\mathcal{K}^\ast_{\mathsf m}(\mathcal D)+\mathrm{L}_xG^{-1}(\epsilon)\big)\geq 1-3\beta.
			\label{eq:proof1}
	\end{equation}}
	\Ali{Confidence $\beta$ is multiplied by $3$ since the Lipschitz continuity is needed in \eqref{eq:RCP} in three different regions and, hence, we leverage the results in \cite{murali2022scenario} to deal with this issue by multiplying $\beta$ by three.} On the other hand, due to the particular selection of $\hat N$ and $\beta_s$ according to Theorem~\ref{theo:empirical}, \payam{we know that \eqref{eq:feasibility123} holds.} Therefore,
	\begin{equation}
		\label{eq:optimal_compare}
		\pr\left(\mathcal{K}^\ast_{\mathsf m}(\mathcal D)\le \mathcal K^\ast(\mathcal D)\right)\ge 1-\beta_s.
	\end{equation}
	\Ali{Define the events
		$\mathcal A  := \{\mathcal D\,|\, \mathcal{K}^\ast\le \mathcal{K}^\ast_{\mathsf m}(\mathcal D)+\mathrm{L}_xG^{-1}(\epsilon)\}$,
		$\mathcal B := \{\mathcal D\,|\,\mathcal{K}^\ast_{\mathsf m}(\mathcal D)\le \mathcal K^\ast(\mathcal D)\}$,
		and
		$\mathcal C :=\{\mathcal D\,|\,\mathcal K^\ast(\mathcal D)+\mathrm{L}_xG^{-1}(\epsilon)\le 0\}$,
		where $\pr(\mathcal A)\ge 1-3\beta $ and  $\pr(\mathcal B)\ge 1-\beta_s$.
		The inequalities in $\mathcal A$ and $\mathcal B$ satisfy
		\begin{equation}
			\label{eq:chain}
			\mathcal{K}^\ast\leq\mathcal{K}^\ast_{\mathsf m}(\mathcal D)+\mathrm{L}_xG^{-1}(\epsilon)\le \mathcal K^\ast(\mathcal D)+\mathrm{L}_xG^{-1}(\epsilon).
		\end{equation}
		Note that any element $\mathcal D$ that belongs to $\mathcal C$ will make the right-hand side of \eqref{eq:chain} non-positive. In addition, if this element also belongs to $\mathcal A\cap \mathcal B$, the two inequalities in \eqref{eq:chain} will also hold, and we get $\mathcal{K}^\ast\le 0$. 
		%
		%Take any element $\mathcal D\in \mathcal C$. 
		%Suppose the set $\mathcal A\cap \mathcal B\cap \mathcal C$ is not empty. Take any element $\mathcal D$ from this set. Then $\mathcal D\in \mathcal C$, which makes the right-hand side of \eqref{eq:chain} non-positive. This results in $\mathcal{K}^*\le 0$.
		%
		\begin{align*}
			& \pr(\mathcal K^\ast\le 0)
			%
			%\ge \pr(\mathcal A\cap \mathcal B\cap \mathcal C\neq\emptyset\,|\,\mathcal C\neq\emptyset)\\
			\ge  \pr(\mathcal A\cap \mathcal B)
			\geq 1-\pr(\mathcal A^c)-\pr(\mathcal B^c)\nonumber
			\geq 1-3\beta-\beta_s.
		\end{align*}
		%
		%It belongs to $\mathcal A\cap \mathcal B$ with a probability of at least $1-3\beta-\beta_s$:
		%\begin{align}
		%\pr(\mathcal A\cap \mathcal B)\geq 1-\pr(\mathcal A^c)-\pr(\mathcal B^c)\nonumber
		%\geq 1-3\beta-\beta_s.
		%\end{align}
		%Therefore, with a probability of at least $1-3\beta-\beta_s$, $\mathcal D$ belongs to $\mathcal A\cap \mathcal B\cap \mathcal C$, hence $\mathcal A\cap \mathcal B\cap \mathcal C\neq\emptyset$, and in turn $\mathcal K^\ast\le 0$.
		This completes the proof since non-positiveness of $\mathcal K^\ast$ ensures a safety lower bound $(1-\rho)$ with confidence of at least $1-3\beta-\beta_s$.}
\end{proof}
%\begin{theorem}
%	\AL{Abolfazl: omit the theorem and replace it by Proposition~\ref{Thm2}. These two are indeed equivalent, but since the justification of Proposition~\ref{Thm2} is much more easy and helpful for our setting, we should stick to it.} Let us consider the unknown system which is characterized in \eqref{eq:mainsystem} and a uniform random sampling $\omega_N$ of that as in \eqref{eq:data}.  Given a uniform level-set bound $h(\epsilon)$ and $\epsilon,\beta\in[0,1]$, for $N\geq N(\epsilon,\beta)$ we have
%		\begin{align}
%	&\pr_s^N\Big(\pr(\mathcal{S}\models \mathtt{Safety}\mid\omega_N)\ge1-\Delta,\nonumber\\&\Delta-\Delta^*\in[0,\mathrm{I}(\epsilon)]\Big)\geq1-\beta-\mathrm{I}_s, \quad\forall x_0\in X_{nit},
%	\end{align}
%	where
%	\begin{align}
%\mathrm{I}(\epsilon)=min\Big\{h(\epsilon),\;\;\underset{\gamma,c,\lambda}{max}(\Delta)-\underset{\gamma,c,\lambda}{min}(\Delta)\Big\}.
%	\label{eq:bound}
%	\end{align}
%where $\Delta=\frac{\gamma+cT_d}{\lambda}$ is the optimal value of the sampled problem in \eqref{eq:sampledoptima} and $\Delta^*$ is the optimal value of the original optimization problem in \eqref{eq:originaloptima}.
%\label{theo:peyman}
%	\end{theorem}
%\vspace{.25em}

\Ali{
	\begin{corollary}
		\label{rem:gofepsilon}
		If samples are collected uniformly from a hyper rectangular state set with edges of length $\eta_x(i)$ in each dimension $i$, then one can compute $G(\epsilon)$ as $\frac{a\epsilon^n}{\prod_{i=1}^{n} \eta_x(i)}$
		, where $a=\dfrac{1}{2^n}\frac{\pi^{\frac{n}{2}}}{\Gamma(\frac{n}{2}+1)}$ with the Gamma function defined as $\Gamma(k) = 1\times 2\times 3\ldots\times(k-1)$ and $\Gamma(k+\frac{1}{2})=\frac{1}{2}\times \frac{3}{2}\times\ldots(k-\frac{3}{2})(k-\frac{1}{2})\pi^{\frac{1}{2}}$ for all positive integers.
	\end{corollary}
}
\Ali{
	\begin{corollary}
		If the state set is an n-dimensional hypersphere with radius $\tilde{r}$ and the data is sampled uniformly, then one has
		\begin{align*}
			G(\epsilon)=\frac{1}{2}\Big[\mathrm I(1-\frac{c_1^2}{\tilde{r}^2};\frac{n+1}{2},\frac{1}{2})+\frac{\epsilon^n}{\tilde{r}^n}\mathrm I(1-\frac{c_2^2}{\epsilon^2};\frac{n+1}{2},\frac{1}{2}\Big],
		\end{align*}
		where $c_1=\frac{2\tilde{r}^2-\epsilon^2}{2\tilde r}$, and $c_2=\frac{\epsilon^2}{2\tilde r}$.
		\label{rem:gofepsilon2}
	\end{corollary}
}

\Ali{
	\begin{remark}
		For uniform sampling, the function $G(r)$ is proportional to $r^n$. Therefore, the sample complexity of the proposed approach is in the order of $(\frac{v\mathrm{L}_x}{\epsilon})^n$, where $v$ is the volume of state set and $n$ is the dimension of the state set.
	\end{remark}
}

\begin{remark}
	The barrier function constructed based on the finite number of samples according to the above theorem together with the obtained parameters $c$ and $\lambda$ satisfies the conditions \eqref{eq:bar1}-\eqref{eq:bar3} in Definition~\ref{def:barrier} with a confidence of at least $1-3\beta-\beta_s$.
	\label{remark:tooriginal}
\end{remark}
\begin{remark}
	Note that the constraint $g_4$ in \eqref{eq:RCP} enforces the constraint  $\pr(\mathcal{S}\models_{\mathcal{H}} \Psi)\geq 1-\rho$ for a given $\rho$. When $\rho$ is not fixed, one can eliminate this constraint from the optimization and guarantee directly the following inequality
	\begin{align*}
		\pr_w(\mathcal{S}\models_{\mathcal{H}} \Psi)\geq 1-\frac{1+c^*\mathcal{H}}{\lambda^*},
	\end{align*}
	where $c^*$ and $\lambda^*$ are the optimal values of the SCP\textsubscript{$\scriptscriptstyle N,\hat{N}$}. 
	This increases the likelihood of getting a feasible optimization and gives the best possible lower bound on the safety probability. 
	%This perspective may result in less conservativeness in some situations that we are not interested in verifying a special lower bound of probability for safety of the system.
	\label{remark:non_deterministic_verification}
\end{remark}

For the sake of clarity, we present the steps required for applying Theorem~\ref{theo:peyman} in Algorithm~\ref{alg:verification}.

\begin{algorithm}
	\caption{Safety verification of an unknown dt-SS $\mathcal{S} = (X,V_w,w,f)$ using collected data.}
	\SetAlgoLined
	\Ali{
		\textbf{Input:} Confidence parameters $\beta\in[0,1]$ and $\beta_s\in [0,1)$, parameters $\rho \in(0,1]$, $\delta \in \mathbb{R}^+$, $\hat{M} \in \mathbb{R}^+$, $\mathrm{L}_x \in \mathbb{R}^+$, and the degree of barrier certificate $\mathcal{Q}$\\
		\textbf{1:} Compute the number of samples $\hat{N} \ge \hat M/(\delta^2\beta_s)$ to be used for the empirical average (Theorem~\ref{theo:empirical})\\ 
		\textbf{2:} Choose the number of samples $N$\\
		\textbf{3:} Compute $\epsilon=\mathrm I ^{-1}(1-\beta;\mathcal Q+3,N-\mathcal{Q}-2)$\\
		\textbf{4:} Select a probability measure $\pr$ for the \payam{state set} $X$\\
		\textbf{5:} Collect $N\hat N$ state pairs from the system
		$$\payam{\mathcal D = \{(x_i,x^+_{ij})\in X^2,\,\, x^+_{ij} = f(x_i,w_{ij})\}_{i,j}}$$\\
		\textbf{6:} Solve SCP\textsubscript{$\scriptscriptstyle N,\hat{N}$} in \eqref{eq:SCN2} with $\mathcal D$ and obtain the optimal solution $\mathcal{K}^*(\mathcal{D})$
		\\
		\textbf{Output:} 
		If  $\mathcal{K}^*(\mathcal{D})+\mathrm{L}_xG^{-1}(\epsilon) \leq 0 $, then $\pr_w(\mathcal{S}\models_{\mathcal{H}} \Psi)\ge1-\rho$ with a confidence of at least $1-3\beta-\beta_{s}$.
		\\
		\label{alg:verification}}
\end{algorithm}

%placeholder
%\textcolor{red}{what do w do if $\mathcal{K}^*(\mathcal{D})+\epsilon > 0$? Perhaps we need a remark to discuss this.} \textcolor{orange}{I thought about this comment. I think that we cannot say anything about if this inequality does not hold. It's just like a sufficient condition.}
\begin{comment}
\begin{remark}
\label{rem:barrier}
Theorem~\ref{theo:peyman} provides a connection between
the optimal values of SCP\textsubscript{$\scriptscriptstyle N,\hat{N}$} and that of the original RCP in \eqref{eq:RCP}, and accordingly, provides a lower bound \payam{for the}
probability \payam{of} safety \payam{for} an unknown system with a
confidence of at least $1-3\beta-\beta_s$. According to \cite[Lemma 3.2]{esfahani2014performance}, one can improve  the confidence to $1-\beta_s$ by making all the constraints of SCP\textsubscript{$\scriptscriptstyle N,\hat{N}$} \Ali{more conservative with quantity $\mathrm{L}_{x}G^{-1}(\epsilon)$}. Although this will improve the confidence, it will lead to very conservative constraints and \Ali{a lower likelihood of getting $\mathcal{K}^*(\mathcal{D})+G^{-1}(\epsilon) \leq 0$.}
% In this case, the obtained barrier function is a BC for the original stochastic system with a confidence of at least $1-\beta_s$, namely $\beta=0$.
\end{remark}
\end{comment}

%\subsection{Computation of the Lipschitz Constant} 
%\label{sub:uniform}
Both Theorem~\ref{theo:peyman} and Algorithm~\ref{alg:verification} require knowing an upper bound for Lipschitz constant $\mathrm{L}_x$.
The following lemma shows how to get this constant for quadratic barrier certificates and systems with additive noises.
%we assume that the candidate function for barrier certificate is in the quadratic form. This assumption is not restrictive since any 
A similar reasoning can be used for other polynomial-type barrier certificates by casting them as quadratic functions of monomials.
\vspace{.2cm}
\begin{lemma}
	Consider a nonlinear system with additive noise
	\begin{align}
		x(t+1)=f_a(x(t))+w(t),\quad t\in\mathbb N_0,
		\label{eq:additive}
	\end{align}
	and a bounded state set $X$ such that $||x||\leq \mathcal{L}$ for all $x\in X$.
	Without loss of generality, we assume that the mean of noise is zero.
	%	where $f_a(x)$ is the deterministic part of $f(x(t),w(t))$. 
	Let $||f_a(x)||\leq L_1||x||+L_2$ and $||\mathbf{J}_x||\leq \hat{L}$ for some $L_1,L_2,\hat L\ge 0, \forall x \in X$, where $\mathbf{J}_x$ is the Jacobian matrix of $f_a(x)$. Given a quadratic barrier function $x^T\mathrm{P}x$ with a symmetric \payam{positive definite} matrix $\mathrm{P}$, the Lipschitz constant $\mathrm{L}_{x}$ can be upper-bounded by 
	%$2\mathcal{L}(L\hat{L}+1)\|\mathrm{P}\|$. 
	$$2\|\mathrm{P}\|(L_1\mathcal{L}\hat L+ L_2 \hat L + \mathcal{L}).$$
	\label{lemma:nonlinearlemma}
\end{lemma}

%\Sadegh{The proof below was not entirely correct. I had to revise the steps.}
\begin{proof}
	We first compute the Lipschitz constant of $g_5$ in \eqref{eq:bar3} as
	\begin{equation*}
		L_{x_5}=\max\left\{\left\Vert\frac{\partial{g_5(x)}}{\partial x}\right\Vert,\, x\in X,\; \Vert x\Vert \leq \mathcal{L}\right\},
	\end{equation*}
	where
	\begin{align*}
		g_5(x)\!=&\EE\big[(f^T(x(t))+w^T(t))\mathrm{P}(f(x(t))+w(t))\big]\\
		&-x^T(t)\mathrm{P}x(t) -c\\
		=&f^T(x(t))\mathrm{P}f(x(t)) \!-\! x^T(t)\mathrm{P}x(t) \!+\! \EE\big[w^T(t) \mathrm{P}w(t)\big] \!-\!c. 
	\end{align*}
	%Since $\EE[w(t)]=0$, and
	By considering $\mathbf{J}_x=[\frac{\partial f}{\partial x_1},\ldots,\frac{\partial f}{\partial x_n}]$, one has
	\begin{align*}
		L_{x_5}& =\underset{x}{\max}\Vert2(f(x(t))^T\mathrm{P}\;\mathbf{J}_x-x^T(t)\mathrm{P})\Vert\nonumber\\
		&\leq\underset{x}{\max}\;\;2 \| f(x(t))^T\| \|\mathrm{P}\| \|\mathbf{J}_x\| + 2\Vert x^T(t)\Vert\|\mathrm{P}\| \\
		& \leq 2(L_1\mathcal{L}+L_2)\|\mathrm{P}\|\hat{L}+ 2\mathcal{L}\|\mathrm{P}\|\\
		& = 2\|\mathrm{P}\|(L_1\mathcal{L}\hat L+ L_2 \hat L + \mathcal{L}).
	\end{align*}
	Similarly, one can readily deduce that $\mathrm{L}_{x_1}=\mathrm{L}_{x_2}=\mathrm{L}_{x_3}=2\mathcal{L}\|\mathrm{P}\|$, and $\mathrm{L}_{x_4}=0$. Then $\mathrm{L}_{x}=\max(\mathrm{L}_{x_1},\mathrm{L}_{x_2},\mathrm{L}_{x_3},\mathrm{L}_{x_4},\mathrm{L}_{x_5})=2\|\mathrm{P}\|(L_1\mathcal{L}\hat L+ L_2 \hat L + \mathcal{L})$, which completes the proof.
\end{proof}

\begin{remark}
	Note that according to the above lemma, computing the upper bound for Lipschitz constant $\mathrm L_x$ depends on $\|\mathrm{P}\|$. On the other hand, computing the entries of $\mathrm{P}$ depends on Lipschitz constant $\mathrm{L}_x$. In order to tackle this circulatory issue, we consider an upper bound for $\|\mathrm{P}\|$ and enforce it as an additional constraint while solving the SCP in \eqref{eq:SCN2}. If there is no solution with the selected upper bound, we iteratively increase the upper bound until we find a solution or a predefined maximum number of iterations is reached.
\end{remark}
%Next remark  states that an upper bound can be also found for a system in \eqref{eq:mainsystem} which is characterized by linear dynamics with similar reasoning.

\begin{remark}
	\label{rem:linearlip}
	If the underlying dynamics is affine in the form of $x(t+1)=Ax(t)+B+w(t)$ with $A\in \mathbb R^{n\times n}$ and $B\in \mathbb R^{n\times 1}$, we can set $L_1=\hat L$ as an upper bound on $||A||$ and $L_2$ as an upper bound on $\|B\|$.
	%\Sadegh{what about if the system has a bias in the form of $x(t+1)=Ax(t)+B+w(t)$. This bias could be related to the mean of the noise.} 
	%\payam{We considered $w(t)$ to be a Gaussian noise with zero mean. Does not make sense yet?}
	%	\textcolor{red}{Notice that I had generalised the statement and eliminated the assumption of being Gaussian. It could have any distribution. I was also thinking  that we should eliminate the assumption of having zero mean. That's why I wrote "without loss of generality". For this we need to allow the function  $f_a$ to be non-zero, but the assumption $\|f_a(x)\|\le L\|x\|$ enforces that $f_a(0) = 0$, which is not good. See above how I have generalised this. Please modify the case study if needed. I have replaced $L$ with $L_1$ and $L_2$.}
	%	one can employ a similar argument as in Lemma~\ref{lemma:nonlinearlemma} and compute an upper bound for the Lipschitz constant $\mathrm{L}_{x}$ as $2\mathcal{L}(L^2+1)\fancyR(\mathrm{P})$. 
\end{remark}

\begin{remark}
	\label{rem:data_lip}
	\Ali{The Lipschitz constant in Assumption~\ref{assum:lip} can also be estimated directly from the data using Extreme Value Theory with the estimation approach described in \cite{wood1996estimation}.
		%The last constraint of RCP , i.e., $g_5$ in \eqref{eq:sampledcondition} has the maximum Lipschitz constant according to the proof of Lemma~\ref{lemma:nonlinearlemma}.
		For instance, to estimate the Lipschitz constant of $g_5$ in \eqref{eq:sampledcondition}, we gather data $\left\{(x_{i1},x_{i2})\,|\, i_1, i_2 = 1,\ldots,\tilde{N}\right\}$ and compute 
		\begin{align}
			\hat{\mathrm{L}} =  \max \frac{\norm{g_5(x_{i_1})-g_5(x_{i_2})}}{\norm{x_{i_1}-x_{i_2}}}, \qquad i_1, i_2 \in \{1,\ldots,\tilde{N}\}.
			\label{eq:data_lip}
		\end{align}
		The Lipschitz constant of $g_5$ is computed by fitting a Reverse Weibull distribution to the samples of the random variable $\hat{\mathrm{L}}$, and then computing the location parameter of that distribution.
		%can be computed by fitting a Reverse Weibull distribution over $\tilde{M}$ realizations of the following expression for $j \in \{1,\ldots,\tilde{M}\}$ :
		%where $\tilde{N}$ is the number of sample pairs $(x_{i_1},x_{i_2})\in X$ at each realization.
	}
\end{remark}

%Note that in practice after a specific number of samples, coefficients of the obtained BC slightly changes when the sample number increases. Therefore, this obtained barrier certificate is an appropriate candidate to compute the Lipschitz constant as in \eqref{eq:data_lip}.

\vspace{.2cm}

\section{Data-Driven Controller Synthesis}
\label{sec:synthesis}
% In the previous sections, we proposed methodologies on the safety verification of unknown stochastic systems. Here, we shift the emphasis from the verification to the synthesis which means we would like to design controllers rendering safety property for unknown systems.
In this section, we study the problem of synthesizing a controller for an unknown \payam{stochastic control} system using data to satisfy safety specifications. Our approach is to use \emph{control barrier certificates}, fix a parameterized set of controllers, and design the parameters using an SCP.
The stochastic control system is defined next.
% We provide an approach in this section which synthesizes a controller for a black-box system through collected data. It is assumed that inputs live in an infinite set and a fixed structure is considered for the controller. This structured controller is incorporated into a robust convex program founded based on a control safety problem which is introduced later in this section.
% We deal with a stochastic (potentially) unknown discrete-time continuous-space control system which is defined next.
\begin{definition}
	\label{def:mainsys2}
	A discrete-time stochastic control system (dt-SCS) is a tuple $\Ali{\mathcal{S}} = (X,U,V_w,w,f)$, where $X,V_w,w$ are as in Definition~\ref{def:mainsys}, $U\subset \mathbb{R}^m$ is the \payam{input set}, and $f:X\times U\times V_w\rightarrow X$ is the state transition map. The evolution of the state is according to equation
	\begin{align} 
		\Ali{\mathcal{S}}: x(t+1)=f(x(t),u(t),w(t)), \;t\in\mathbb{N}_{0}.
		\label{eq:mainsystem2}
	\end{align}
	%where $u$ is the input vector and $u=(u_1,\ldots,u_m)\in U \subseteq \mathbb{R}^m$. Other notations are the same as explained in Definition~\ref{def:mainsys}. 
\end{definition}
\smallskip

We assume that the map $f$ and distribution of $w$ is unknown but we can gather \payam{data $(x_i, u_i, x^+_i)$} by initializing the system at $x_i$, applying the input $u_i$, and observing the next state of the system $x^+_i=x_i(t+1)$. The collected dataset is
\begin{equation}
	\label{eq:data2}
	\Ali{\mathcal{D}}:=\Big\{(x_i,u_i,f(x_i,u_i,w_j))\Big\}_{i,j}\subset X \times U\times X.
\end{equation}
%\Sadegh{give the following sentence in a formal Problem environment. This section doesn't  have a problem statement.}
%\vspace{1.5mm}
\payam{Now, we state the main problem we are interested to solve here.}
\begin{problem}
	\label{prob:synthesis}
	Consider an unknown dt-SCS $\Ali{\mathcal{S}}$ as in Definition~\ref{def:mainsys2}, with a safety specification $\Psi$ specified by the initial set $X_{in}$, unsafe set $X_u$, and time horizon $\mathcal{H}$. Using a dataset $\Ali{\mathcal{D}}$ of the form \eqref{eq:data2},
	find a \payam{controller $\mathrm{k}:X\rightarrow U$} together with a constant $\rho \in [0,1)$ and confidence $(1-\beta)\in [0,1]$ such that $\Ali{\mathcal{S}}$ under this controller satisfies $\Psi$ with \payam{a probability of at least} $(1-\rho)$, \emph{i.e.,}
	\begin{equation*}
		\pr_w^{\mathrm{k}}\big(\Ali{\mathcal{S}}\models_{\mathcal{H}} \Psi\big)\ge 1-\rho,  \quad\forall x(0)\in X_{in},
	\end{equation*}
	with a confidence $1-\beta$. Moreover, establish a connection between the required size of $\Ali{\mathcal{D}}$ and the confidence $1-\beta$.
\end{problem}

%Our goal in this section amounts to find a controller $u$ to render the safety specification for a dt-SCS described in the above definition within a time horizon $\mathcal{H}$ based on collected data. We denote collected data from the unknown system $f$ by
%\begin{align*}
%\mathcal{D}=\Big\{\hat{x}_i,\hat{u}_i,f(\hat{x}_i,\hat{u}_i)\Big\}_{i=1}^{N}\subseteq X \times U,
%\end{align*}
%where $N$ is the number of samples.

%\Sadegh{Revise algorithm 2 similar  to what i did with algorithm 1.}

%\Sadegh{Revise Lemma 5.9 similar to what I  did in the previous section.}

%\Sadegh{Revise Theorem 5.7 similar to what I did in the previous section.}

%Inspired by \cite{jagtap2019formal}, we adapt the safety control problem for our setting in the next lemma. It utilizes the concept of control barrier certificate to synthesize a controller such that it forces the stochastic system to be safe with a probability lower bounded by a certain value.
Similar to the verification problem discussed in the previous sections, we use the notion of control barrier certificates with a parameterized set of controllers \cite{jagtap2019formal} to get a characterization of the controller together with the lower bound on the safety probability.
\begin{definition}
	Given a dt-SCS $\mathcal S = (X,U,V_w,w,f)$ with $U\subset\mathbb R^m$, initial set $X_{in}\subset X$, and unsafe set $X_u\subset X$, a function $\Ali{\mathrm{B}}:X\rightarrow\mathbb{R}_{0}^{+}$ is called a control barrier certificate (CBC) for $\mathcal S$ if there exist constants $\lambda > 1$, $c\ge 0$, \Ali{and functions $\mathscr{P}_{\ell}(x):X\rightarrow\mathbb{R}_{0}^{+}$, $\ell\in\{1,2,\ldots,m\}$, such} that \Ali{constraints in \eqref{eq:bar1} and \eqref{eq:bar2} hold, and
		\begin{align}
			&\mathbb{E}\Big[\Ali{\mathrm{B}}(f(x,u,w))\mid x,u\Big]+\sum_{\ell=1}^{m}(u_\ell-\mathscr{P}_{\ell}(x))\leq \Ali{\mathrm{B}}(x)+c\nonumber \\
			&\qquad\qquad\qquad\forall x \in X,\; \forall u = \payam{[u_1;\ldots;u_m] \in U.}
			\label{eq:baru33}
		\end{align}
		\label{def:barrier_cc}}
\end{definition}

\begin{theorem}
	\label{thm:CBS}
	A CBC $\Ali{\mathrm{B}}(x)$ as in Definition~\ref{def:barrier_cc} guarantees that
	\begin{equation*}
		\pr_w^{\mathrm{k}}\big(\mathcal{S}\models_{\mathcal{H}} \Psi\big)\ge 1-\rho,  \quad\forall x(0)\in X_{in},
	\end{equation*}
	\payam{under the controller $\mathrm{k}(x) = [\mathscr{P}_1(x);\mathscr{P}_2(x);\ldots;\mathscr{P}_m(x)]$, where $\rho = (1+c\mathcal{H})/\lambda$} with $\mathcal{H}$ being the time horizon of the safety specification.
\end{theorem}
%$\mathrm{B}_u(x)$ and $\mathscr{P}_\ell(x)$ to be

\medskip

\Ali{Let us consider polynomial-type CBC and controllers. The number of CBC coefficients is denoted by $\mathcal{Q}$.}
%$\mathrm{B}_u(x)$ will be of the form \eqref{eq:fixedbarrier}.
Polynomial $\mathscr{P}_\ell$ has the following form for some $k'\in\mathbb{N}_0$: 
%each controller $u_{\ell}$ as a member of the input vector $u=(u_1,u_2,\ldots,u_m)\in U \subset \mathbb{R}^m$ and in the  form of a $k'$-fixed degree polynomial function $k'\in\mathbb{N}_0$ as follows:
\begin{align}
	\mathscr{P}_\ell(p^\ell,x)=\sum_{\iota_1=0}^{k'}\ldots\sum_{\iota_n=0}^{k'}p^\ell_{\iota_1,\ldots,\iota_n}(x_1^{\iota_1}\ldots x_n^{\iota_n}),
\end{align} 
with $p^\ell_{\iota_1,\ldots,\iota_n}=0$ for $\iota_1+\ldots+\iota_n>k'$.

The overall number of all coefficients of $m$ polynomials $\mathscr{P}_\ell(p^\ell,x)$ is denoted by $\mathcal{P}$.
%We encapsulate the boundedness of the input as a polytope described by
%\begin{align}
%\mathcal{A}\mathsf{P}\leq\mathsf{b},
%\label{eq:input_boundedness}
%\end{align}
%where $\mathsf{P}=[\mathscr{P}_{u_{1}};\ldots;\mathscr{P}_{u_m}]$, $\mathsf{A}\in\mathbb{R}^{m \times m}$, and $\mathsf{b} \in \mathbb{R}^{m \times 1}$. 
We also assume that the input set $U$ is a polytope of the form
\begin{equation}
	U  =  \left\{u\in\mathbb R^m\,|\,\mathcal{A}u\leq\mathsf{b}\right\},
	\label{eq:input_boundedness}
\end{equation}
for some $\mathcal{A}\in\mathbb{R}^{q \times m}$ and $\mathsf{b} \in \mathbb{R}^{q \times 1}$.

Under these assumptions, the inequalities in Definition~\ref{def:barrier_cc} and Theorem~\ref{thm:CBS} can be written as an RCP:
%
%	\begin{remark}
%	The control barrier certificate is also supposed to be a polynomial function with degree of $k\in \mathbb{N}_0$ as 
%	\begin{align}
%	\mathrm{B}_{u}(b,x)=\sum_{\iota_1=0}^{k} \ldots\sum_{\iota_{k}=0}^{k}b_{\iota_1,\ldots,\iota_{k}}(x_1^{\iota_1}\ldots x_{k}^{\iota_{k}}),
%	\label{eq:barpoll}
%	\end{align} 
%	with $b_{\iota_1,\ldots,\iota_{k}}=0$ for $\iota_1+\ldots+\iota_{k}>k$. Number of coefficients of this barrier certificate is denoted by $\mathcal{Q}$.
%	\label{rem:barrier6}
%\end{remark}
%As the first step towards the developing of our proposed synthesis method, we reformulate the safety control problem in Lemma~\ref{def:barrier_cc} as a robust convex problem with universal quantifiers on both state and input sets as following:
\begin{align} 
	\text{RCP}:\left\{
	\begin{array}{ll}
		\underset{d}{\min}\quad \mathcal{K}\quad \quad\\
		\text{s.t.}\quad\;
		\max_z \;\;g_{z}(x,u,d) \leq 0,\\
		\qquad \;\;\;z \in \{1,2,\ldots,5+q\}, \forall x \in X,\forall u \in U,\\
		\quad \quad\;\;\;d=[\mathcal{K};\lambda;c;b_{\iota_1,\ldots,\iota_{n}};p^\ell_{\iota_1,\ldots,\iota_{n}}],\\
		\qquad ~~\mathcal{K}\in\mathbb{R}, \; \lambda >1, \;c \geq 0,\\
	\end{array}
	\right.
	\label{eq:RCPP}
\end{align}
where \Ali{ $g_z(x,d), z\in\{1,\ldots,4\}$, are the same as \eqref{eq:sampledcondition}, and
	\begin{align}
		&g_5(x,u,d)=\mathbb{E}\Big[\Ali{\mathrm{B}}(b,f(x,u,w))\mid x,u\Big]+\sum_{\ell=1}^{m}(u_\ell-\mathscr{P}_{\ell}(p^\ell,x))\nonumber \\&-\Ali{\mathrm{B}}(b,x)-c-\mathcal{K},\nonumber\\
		&[g_6(x,d);\ldots;g_{5+q}(x,d)]\!=\!\;\mathcal{A}\;[\mathscr{P}_1(p^1,x);\ldots;\!\mathscr{P}_m(p^m,x)]\!\;-\nonumber\\
		&\!\mathsf{b}\!\;-\!\;\mathcal{K}\mathbf{1}_{q\times1}.
		\label{eq:sampledconditionn}
\end{align} }
Note that the last inequality in \eqref{eq:sampledconditionn} encodes the fact that the control input should be inside the set $U$ specified by the polytope \eqref{eq:input_boundedness}.

%in which $\mathbf{g}_4(x,d)=[\cg_{41}(x,d);\ldots;\cg_{4m}(x,d)]$ is a column vector containing functions that are used in encapsulating the boundedness of $m$ inputs according to \eqref{eq:input_boundedness}. The fixed structure of the controller is incorporated into the constraints of the problem as polynomials $\mathscr{P}_{u_{\ell}}(x)$.  The coefficients of this controller, namely $p_{\iota_1,\ldots,\iota_{k'}}$, are considered decision variables of the optimization problem and can be computed by solving the optimization problem. 

\Ali{The constraints} in the RCP is always feasible. A solution can be  constructed as follows.
Set the coefficients of $\Ali{\mathrm{B}(b,x)}$ and $\mathscr{P}_\ell(p^\ell,x)$ equal to zero, $c=0$, $\lambda=2$, and $u_{\ell}=\mathscr{P}_{\ell}(p^\ell,x)\;\forall \ell\in\{1,\ldots,m\}$. Also select $\mathcal{K}$ large enough such that $\mathcal{K}\geq \frac{1}{\rho}-2$ together with $\mathcal{K}\;\mathbf{1}_{m \times 1}\geq -\mathsf{b}$.

\payam{The RCP in \eqref{eq:RCPP} is in general hard to solve since the map $f$ and the probability measure $\pr_w$ are unknown.} Hence, similar to the verification approach discussed in Section~\ref{sec:datadriven}, \payam{we assign a probability distribution to both state and input sets}, and collect $N$ i.i.d pairs $(x_i,u_i)$ from this assigned distribution, and replace the robust quantifiers $\forall{x} \in X$ and $\forall{u} \in U$ with $\forall{x_i} \in X$ and $\forall{u_i} \in U, i\in\{1,\ldots,N\}$, respectively. This results in a scenario convex program called SCP\textsubscript{$\scriptscriptstyle N$}, which is not presented here for the sake of brevity.
% since it is needed to present our results towards the end of this section.

To address the issue of unknown $f$ and $\pr_w$, the expectation in $g_5$ is replaced with its empirical approximation by sampling $\hat N$ i.i.d. values $w_j,\; j \in \{1,\ldots,\hat{N}\}$, from $\pr_w$ for each pair of $(x_i,u_i)$, which \payam{results in} the following scenario convex program denoted by SCP\textsubscript{$\scriptscriptstyle N,\hat{N}$}:

\begin{align} 
	\text{SCP\textsubscript{$\scriptscriptstyle N,\hat{N}$}}:\left\{
	\begin{array}{ll}
		\underset{d}{\min}\quad \mathcal{K}\quad \quad\\
		\text{s.t.}\quad\;
		\max_z\; \bar{g}_{z}(x_i,u_i,d) \leq 0,\\~\qquad \; z \in \{1,2,\ldots,5+q\}, \\
		\qquad ~~\forall x_i \in X, \;\forall u_i \in U, \forall i \in \{1,\ldots,N\},\\
		\quad \quad\;\;\;d=[\mathcal{K};\lambda;c;b_{\iota_1,\ldots,\iota_{n}};p^{\ell}_{\iota_1,\ldots,\iota_{n}}],\\
		\qquad ~~\mathcal{K}\in\mathbb{R}, \; \lambda >1, \;c \geq 0,\\
	\end{array}
	\right.
	\label{eq:SCP22}
\end{align}
where $\bar{g}_z:=g_z$ for all $z \in \{1,2,\ldots,5+q\}\setminus\{5\}$, and 
\begin{align}
	&\bar{\cg}_5(x_i,u_i,d)=\frac{1}{\hat{N}}\sum_{j=1}^{\hat{N}}\Ali{\mathrm{B}}(b,f(x_i,u_i,w_j))\;+\nonumber\\
	&\sum_{\ell=1}^{m}(u_{{i}_\ell}-\mathscr{P}_{\ell}(p^\ell,x_i))-\;\Ali{\mathrm{B}}(b,x_i)-c+\delta-\mathcal{K}.
	\label{eq:empiricall}
\end{align}
Using empirical approximation introduces an error which is demonstrated by $\delta$ in the above optimization problem.
We denote by $\hat{\mathrm{B}}_u(b,x\,|\,\Ali{\mathcal{D}})$ the constructed control barrier certificate with coefficients computed by solving the SCP\textsubscript{$\scriptscriptstyle N,\hat{N}$}.
\begin{remark}
	Similar to Theorem~\ref{theo:empirical}, under the assumption
	$$\text{Var}\big(\Ali{\mathrm{B}}(b,f(x,u,w))\big)\leq \hat{M},$$
	for some $\hat{M}>0$,
	a desired confidence $\beta_s\in(0,1]$, and an error $\delta$, one has
	\begin{align}
		\pr_w^{\mathrm{k}} \Big(\hat{\mathrm{B}}_u(b,x\mid\Ali{\mathcal{D}}) \models\text{SCP\textsubscript{$\scriptscriptstyle N$}}\Big)\geq 1- \beta_s,
		\label{eq:feasibilityy}
	\end{align}
	provided that $\hat{N} \geq \frac{\hat{M}}{\delta^2\beta_s}$.
	\label{remark:empirical_inputt}
\end{remark}
\Ali{To provide the main results here, we need the following assumptions.}
\begin{assumption}
	\label{assum:lipp}
	Function $g_{5}$ is Lipschitz continuous with respect to $(x,u)$ with Lipschitz constant $\mathrm{L}_{5}$. 
	\Ali{Functions $g_1,g_2,g_3,g_6,\ldots, g_{5+q}$ are also Lipschitz continuous with respect to $x$ with Lipschitz constants $\mathrm{L}_1,\mathrm{L}_2,\mathrm{L}_3,\mathrm{L}_6,\ldots,\mathrm{L}_{5+q}$, respectively. Then, the Lipshitz constat of maximum of these function is $\mathrm{L}_1+\mathrm{L}_2+\mathrm{L}_3+\mathrm{L}_{5}+\mathrm{L}_6+\ldots+\mathrm{L}_{5+q}$. Furthermore, if all functions $g$ are analytic over a compact domain $X \times U$, the Lipschitz constant of their maximum is $ \max (\mathrm{L}_1,\mathrm{L}_2,\mathrm{L}_3,\mathrm{L}_{5}, \mathrm{L}_6,\ldots,\mathrm{L}_{5+q}$), which we denote it by $\mathrm L_{x,u}$.}
\end{assumption}
\Ali{
	\begin{assumption}
		\label{ass:G2}
		There is a strictly increasing function $G(r):\mathbb{R}^+ \rightarrow [0,1]$ such that
		\begin{align}
			\pr[\mathrm{b}(x,u,r)]\geq G(r)\qquad \forall (x,u) \in X \times U,
			\label{eq:gofr2}
		\end{align}
		where $\mathrm{b}(x,u,r)$ is an open ball in the product space $X\times U$ centered at the point $(x,u)$ with radius $r$.
	\end{assumption}
}
Now, we have all the ingredients to propose the main results here.
\begin{theorem}
	Consider an unknown dt-SCS as in Definition~\ref{def:mainsys2} and a safety specification $\Psi$. Let \Ali{Assumptions~\ref{assum:lipp}--\ref{ass:G2} hold with constant $\mathrm{L}_{x,u}$ and function $G(r)$.} Suppose that $\mathcal{K}^*(\Ali{\mathcal{D}})$ is the optimal value of SCP\textsubscript{$\scriptscriptstyle N,\hat{N}$} in \eqref{eq:SCP22} \Ali{with number of samples $N$, a given $\rho \in (0,1]$, and} for $\hat{N}$ selected based on Remark~\eqref{remark:empirical_inputt} with confidence of $1-\beta_s$. \Ali{
		Suppose
		\begin{align}
			\mathcal{K}^*(\mathcal{D})+\mathrm L_{x,u}G^{-1}(\epsilon)\leq 0,
			\label{eq:main_condition_2}
		\end{align}
		where function $G$ is defined in \eqref{eq:gofr2} and $\epsilon=\mathrm I ^{-1}(1-\beta;\mathcal Q+\ \mathcal P+3,N-\mathcal{Q}-\mathcal P-2)$ with confidence parameter $\beta \in [0,1]$, and $\mathcal{Q}$ and $\mathcal{P}$ being respectively the number of coefficients of the polynomial control barrier certificate and the overall number of coefficients of polynomials $\mathscr{P}_{\ell}(p^\ell,x)$ for $m$ inputs.
		Then, the following statement is valid with a confidence of at least $1-3\beta-\beta_s$: the system $\Ali{\mathcal{S}}$ together with the constructed control input 
		$$\mathrm{k}(x):=[\mathscr{P}_1(p^1,x);\ldots;\mathscr{P}_m(p^m,x)],$$
		for which coefficients $p^\ell, \ell \in\{1,\ldots,m\}$, are obtained from the solution of SCP\textsubscript{$\scriptscriptstyle N,\hat{N}$}, is safe within the time horizon $\mathcal{H}$ with a probability of at least $1-\rho$, i.e.,
		\begin{align}
			&\pr^{\mathrm{k}}_w\big(\Ali{\mathcal{S}}\models_{\mathcal{H}} \Psi \big)\ge1-\rho.
			\label{RCC}
		\end{align}
	}
	\label{theo:peyman_contt}
\end{theorem}

\begin{proof}
	\Ali{The proof is similar to the proof of Theorem~\ref{theo:peyman} by replacing  $\pr_w$ with $\pr_w^\mathrm k$ for the RCP \eqref{eq:RCPP} and its associated SCPs. The function $G(\epsilon)$ is defined as in \eqref{eq:gofr2}.
		The number of coefficients is $\mathcal Q + \mathcal P + 3$ where $\mathcal P$ is the overall number of coefficients of $m$ polynomials defining the controller, which results in the new arguments of the regularized incomplete beta function $\mathrm I$ in the theorem statement.} 
\end{proof}
\begin{comment}
\begin{remark}
%Similar to the inference in Remark~\eqref{remark:tooriginal}, 
The control barrier certificate and the controller constructed based on the finite number of samples according to the \Ali{above theorem together with the obtained optimal parameters $c$ and $\lambda$ satisfy conditions in Definition~\ref{def:barrier_cc} with a confidence of at least $1-3\beta-\beta_s$.}
%placeholder
%\textcolor{red}{The same remark as previous section. I would like to see a proof of such a statement in the previous section.}
\end{remark}
\end{comment}

\begin{corollary}
	\label{rem:gofepsiloninput}
	\Ali{If samples are collected uniformly from a hyper rectangular sets $X$ and $U$, respectively, with edges of length $\eta_x(i)$ and $\eta_u(j)$ in each dimension $i$ and $j$, then one can compute $G(\epsilon)$ as $\frac{a\epsilon^{n+m}}{\prod_{i=1}^{n} \eta_x(i)\prod_{j=1}^{m} \eta_u(j)}$
		, where $a=\frac{1}{2^{n+m}}\frac{\pi^{\frac{n+m}{2}}}{\Gamma(\frac{n+m}{2}+1)}$ with Gamma function defined in Corollary~\ref{rem:gofepsilon}. %$\Gamma(n+m+\frac{1}{2})=\pi^{\frac{1}{2}}(n+m-\frac{1}{2})(n+m-\frac{3}{2})\cdots\frac{1}{2}$.
	} 
\end{corollary}
\proof{The proof is similar to the proof of Corollary \ref{rem:gofepsilon} in \ref{sec:proof_g} based on the new definition of $G(r)$ in Assumption~\ref{ass:G2}}.

\begin{remark}
	%Similar to Remark~\eqref{remark:non_deterministic_verification}, 
	When $\rho$ is not fixed, one can eliminate constraint $g_4$ from \eqref{eq:RCPP} and directly provide the following inequality
	\begin{align*}
		\pr_w^\mathrm k(\Ali{\mathcal{S}}\models_{\mathcal{H}} \Psi)\geq 1-\frac{1+c^*\mathcal{H}}{\lambda^*},
	\end{align*}
	in which $c^*$ and $\lambda^*$ are the optimal solutions of SCP\textsubscript{$\scriptscriptstyle N,\hat{N}$} in \eqref{eq:SCP22}. This increases the likelihood of getting a feasible solution and gives the best possible lower bound on the safety probability for $\Ali{\mathcal{S}}$. \Ali{A schematic overview of our synthesis approach is presented in Fig.~\ref{fig:controller}.}
\end{remark}

\begin{algorithm}
	\SetAlgoLined
	\Ali{
		\textbf{Input:} Confidence parameters $\beta\in [0,1]$, $\beta_s\in (0,1]$,
		parameters $\rho\in(0,1]$, $\delta \in \mathbb{R}^+$, $\hat{M} \in \mathbb{R}^+$, $\mathrm{L}_{x,u} \in \mathbb{R}^+$, degree of the barrier certificate $\mathcal{Q}$, and degree of the polynomial functions for the controller $\mathcal{P}$\\
		\textbf{1:} Compute the number of samples $\hat{N} \ge \hat M/(\delta^2\beta_s)$ for the empirical average (Remark~\ref{remark:empirical_inputt})\\
		%	\textbf{2:} Calculate the number of coefficients of fixed control barrier function $\mathcal{Q}'$ and coefficients of fixed-structure control polynomial functions $\mathcal{P}$ according to Remark~\ref{rem:barrier6} and Assumption~\ref{assum:contbar}, respectively \\
		\textbf{2:} Choose the number of samples $N$\\ 
		\textbf{3:} Compute $\epsilon=\mathrm I ^{-1}(1-\beta;\mathcal Q+\ \mathcal P+3,N-\mathcal{Q}-\mathcal P-2)$\\
		\textbf{4:} Select a probability measure $\pr$ for the state-input set $(X,U)$\\
		\textbf{5:} Collect $N\hat N$ tuples from the system $\mathcal D := \{(x_i,u_i,x'_{ij})\in X\times{U}\times{X},\,\, x'_{ij} = f(x_i,u_i,w_{ij})\}_{i,j}$\\
		\textbf{6:} Solve SCP\textsubscript{$\scriptscriptstyle N,\hat{N}$} in \eqref{eq:SCP22} with $\mathcal D$ and obtain the optimal solution $\mathcal{K}^*(\Ali{\mathcal{D}})$\\
		\textbf{Output:} 
		If  $\mathcal{K}^*(\Ali{\mathcal{D}})+\mathrm L_{x,u}G^{-1}(\epsilon) \leq 0 $, then $\pr_w^{\mathrm{k}}(\Ali{\mathcal{S}}\models_{\mathcal{H}}  \Psi)\ge1-\rho$ with a confidence of at least $1-3\beta-\beta_{s}$ and with the controller $\mathrm{k}(x):=[\mathscr{P}_1(p^1,x);\ldots;\mathscr{P}_m(p^m,x)]$.
		%	\textbf{7:} 
		%	Compute barrier certificate constructed from data $\hat{\mathrm{B}}(b,x)$ using optimal solution of SCP\textsubscript{w} for any sample number greater than $N(\bar{\epsilon},\beta)$.
	}
	\caption{Data-driven synthesis for safety specification on an unknown dt-SCS $\Ali{\mathcal{S}} = (X,U,V_w,w,f)$.}
	\label{alg:synthesiss}
\end{algorithm}

\begin{figure}[ht]
	\centering 
	\includegraphics[width=.5
	\linewidth]{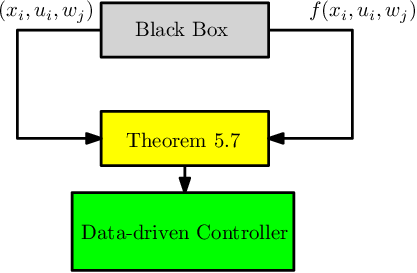}
	\caption{A schematic overview of the data-driven synthesis presented in Section~\ref{sec:synthesis}.}
	\label{fig:controller}
\end{figure}
\payam{Next lemma provides an upper bound for Lipschitz constant $\mathrm{L}_{x,u}$, which is required in Theorem~\ref{theo:peyman_contt}, in the case that the system is affected by an additive noise.} 
\begin{lemma}
	Consider a nonlinear dt-SCS as in Definition~\ref{def:mainsys2} which is affected by an additive noise as the following:
	\begin{align}
		x(t+1)=f_a(x(t),u(t))+w(t),
		\label{eq:additive_u}
	\end{align}
	%	where $f_a(x)$ is the deterministic part of $f(x(t),w(t))$. 
	and a bounded state set $X$ and input set $U$ such that $\|x\|\leq \mathcal{L}_x$ for all $x \in X$, and $\|u\|\leq \mathcal{L}_u$ for all $u \in U$. Without loss of generality, we assume that the mean of the noise is zero. Let $\|f_a(x,u)\|\leq L_1\|x\|+L_2\|u\|+L_3$, $||\mathbf{J}_x||\leq \hat{L}_x$, and  $||\mathbf{J}_u||\leq \hat{L}_u$, for some $\mathcal{L}_x,\mathcal{L}_u,L_1, L_2,L_3,\hat{L}_x, \hat{L}_u \geq 0$, where $\mathbf{J}_x$ and  $\mathbf{J}_u$ are Jacobian matrices of $f_a(x,u)$ with respect to $x$ and $u$, respectively. Given a quadratic barrier function $x^T\mathrm{P}x$, and a set of quadratic functions $x^T\mathrm{P}_\ell x,\; \ell \in \{1,\ldots,m\}$, \payam{representing each} of $\mathscr{P}_\ell(p^\ell,x)$ with symmetric matrices $\mathrm{P}$ and $\mathrm{P}_\ell$, the Lipschitz constant $\mathrm{L}_{x,u}$ can be upper-bounded by  $\sqrt{\mathscr{L}_x^2+\mathscr{L}_u^2}\;$, where
	\begin{align}
		\label{eq:boundss}
		\mathscr{L}_{x}& =2\mathcal{L}_xL_1\hat{L}_x\|\mathrm{P}\|+2\mathcal{L}_uL_2\hat{L}_x\|\mathrm{P}\|+2L_3\hat{L}_x\|\mathrm{P}\|\nonumber\\
		&\;\;\;+\mathcal{L}_x\|\mathrm{P}\|+\mathcal{L}_x\sum_{\ell=1}^{m}\|\mathrm{P}_\ell\|,\\
		\mathscr{L}_{u}
		&=2\mathcal{L}_xL_1\hat{L}_u\;\|\mathrm{P}\|+2\mathcal{L}_uL_2\hat{L}_u\;\|\mathrm{P}\|+2L_3\hat{L}_u\;\|\mathrm{P}\|+\sqrt{m}.\nonumber
	\end{align} 
	\label{lemma:nonlinearlemma2}
\end{lemma}

\begin{proof}
	We first compute the Lipschitz constant regarding $\cg_5(x,u,d)$ in \eqref{eq:sampledconditionn}, 
	where
	\begin{align}
		\cg_5(x,u,d)&=~\EE\big[(f^T(x(t),u(t))+w^T(t))\mathrm{P}(f(x(t),u(t))+\nonumber\\&w(t))\big]+\sum_{\ell=1}^{m}(u_\ell-\mathscr{P}_{\ell}(p^\ell,x))-x^T(t)\mathrm{P}x(t) -c\nonumber. 
	\end{align}
	Considering $\EE[w(t)]=0$, we compute the upper bounds for Lipschitz constant with respect to $x$ and $u$ separately denoted by $\mathrm{L}_{5_x}$
	and $\mathrm{L}_{5_u}$, respectively. We define $\mathbf{J}_x=[\frac{\partial f}{\partial x_1},\ldots,\frac{\partial f}{\partial x_n}]$ and $\mathbf{J}_u=[\frac{\partial f}{\partial u_1},\ldots,\frac{\partial f}{\partial u_m}]$ as Jacobian matrices with respect to $x$ and $u$, respectively. 
	\begin{align*}
		\mathrm{L}_{5_x} =& \max\limits_{x,u}\Vert\frac{\partial{\cg_5(x,u,d)}}{\partial x }\Vert=\underset{x,u}{\max}\;\Vert2(f(x(t),u(t))^T\mathrm{P}\;\mathbf{J}_{x}\nonumber\\& -x^T(t)\mathrm{P}-x^T(t)\sum_{\ell=1}^{m}\mathrm{P}_\ell\Vert\nonumber\\
		\le & 2\mathcal{L}_xL_1\hat{L}_x\|\mathrm{P}\|+2\mathcal{L}_uL_2\hat{L}_x\|\mathrm{P}\|+2L_3\hat{L}_x\|\mathrm{P}\|+\nonumber\\&\mathcal{L}_x\|\mathrm{P}\|+\mathcal{L}_x\sum_{\ell=1}^{m}\|\mathrm{P}_\ell\|,
	\end{align*}
	and accordingly,
	\begin{align*}
		\mathrm{L}_{5_u}
		& =\max\limits_{x,u}\Vert\frac{\partial{\cg_5(x,u,d)}}{\partial u }\Vert\\
		%& =\underset{x,u}{\max}\;\Vert2(f(x(t),u(t))^T\mathrm{P}\mathbf{J}_{u}+\nabla_u\Big(\big(\sum_{\ell=1}^{m}u_\ell-x^T(t)\;\mathrm{P}_\ell \;x(t)\big)\Big)\Vert\\
		& =\Vert2(f(x(t),u(t))^T\mathrm{P}\mathbf{J}_{u}+\mathbf{1}_m\Vert\\
		& \leq 2\mathcal{L}_xL_1\hat{L}_u\;\|\mathrm{P}\|+2\mathcal{L}_uL_2\hat{L}_u\;\|\mathrm{P}\|+2L_3\hat{L}_u\;\|\mathrm{P}\|+\sqrt{m}.
	\end{align*}
	Now it can be deduced that $$\mathrm{L}_5\leq\sqrt{\mathrm{L}_{5_x}^2+\mathrm{L}_{5_u}^2}.$$
	Similar to the proof \payam{of} Lemma~\ref{lemma:nonlinearlemma}, it is straightforward to compute the upper bounds of Lipschitz constants for other constraints in \eqref{eq:sampledconditionn} and show that the computed upper bound is greater than all of them. We ignore this part for the sake of brevity. Then, $\mathrm{L}_{x,u}\leq\max\big(\mathrm{L}_i,i\in\{1,2,\ldots,5+q\}\setminus\{4\}\big)=\sqrt{\mathrm{L}_{5_x}^2+\mathrm{L}_{5_u}^2}$ which is equivalent to $\sqrt{\mathscr{L}_x^2+\mathscr{L}_u^2}$ with $\mathscr{L}_x$ and $\mathscr{L}_u$ \payam{as} in \eqref{eq:boundss}.
\end{proof}
\Ali{Note that one can use similar results as in Remark~\ref{rem:data_lip} to estimate the Lipschitz constant via data.}
\section{Data-driven Barrier Certificates for Non-convex Setting}
\label{sec:datadriven2}
In this section, we extend the proposed result in Section~\ref{sub:main_safe} \payam{to a case of having} non-convex constraints. We modify the constraint \eqref{eq:bar3} in Definition~\ref{def:barrier} as follows:
\begin{align}
	\EE\Big[\mathrm{B}(f(x,w))\mid x\Big]\leq~\kappa\;\mathrm{B}(x)+c,\quad\forall x \in X,
	\label{eq:bar3nonconvex}
\end{align}
where $\kappa \in (0,1)$.

According to the fundamental results in \cite{kushner1967stochastic}, choosing $\kappa$ in the interval $(0,1)$ provides a better lower bound for the probability of safety satisfaction in \eqref{safety_kushi}, namely:
\begin{align*}
	\pr_w\big(\mathcal{S}\models_{\mathcal{H}} \Psi\big)\ge 1-\rho,
\end{align*}
with
\begin{align}
	\rho = \begin{cases}
		1-(1-\frac{1}{\lambda})(1-\frac{c}{\lambda}) \qquad\qquad\qquad\qquad\text{if $\lambda \ge \frac{c}{\kappa}$}\\
		\frac{1}{\lambda}(1-\kappa)^{\mathcal{H}}+\frac{c}{\kappa \lambda}\big(1-(1-\kappa)^{\mathcal{H}}\big)  \;\qquad\;\;\text{if $\lambda < \frac{c}{\kappa}$},
	\end{cases}
	\label{eq:new_bound}
\end{align}
where parameters $c$, $\lambda$, and $\mathcal{H}$ are the same as in Definition~\eqref{def:barrier}.
Another advantage of choosing $\kappa$ in the interval $(0,1)$ is that this new formulation can be \payam{utilized} in the context of compositionality and interconnected systems \cite{Zamani.2017b,SZ.19}.

Replacing the last condition of RCP in \eqref{eq:sampledcondition} with the modified constraint in \eqref{eq:bar3nonconvex} leads to the following optimization problem which is not convex anymore:
\begin{align} 
	\text{RP}:\left\{
	\begin{array}{ll}
		\underset{d}{\min}\quad \mathcal{K}\quad \quad\\
		\text{s.t.}
		\quad \max_z\big(g_z(x,d)\big) \!\leq\! 0, z \!\in\!\{1,\dots,4\},\forall x \!\in\! X,\\
		\quad ~~ \;\;d=[\mathcal{K};\lambda;c;b_{\iota_1,\ldots,\iota_n};\kappa],\\
		\quad ~~ \;\;\mathcal{K}\in \mathbb{R},\; \lambda > 1,\; c \geq 0,\; \kappa \in (0,1),
	\end{array}
	\right.
	\label{eq:RCP_non}
\end{align}
\Ali{in which $g_z(x,d), z \in \{1,2,3\}$, are the same as in \eqref{eq:sampledcondition}, and}
\begin{align}
	&g_4(x,d)=\EE\Big[\mathrm{B}(f(x,w))\mid x\Big]\leq~\kappa\;\mathrm{B}(x)+c,\quad\forall x \in X.
	\label{eq:sampledcondition_non}
\end{align}
%\textcolor{red}{why $\kappa$ is not part of your optimisation variables? Where is the  requirement that $\kappa\in(0,1)$? You shouldn't just copy and  paste things, you should adapt them to the new  formulation.}

The non-convexity comes from the multiplication of $\kappa$ and coefficients of barrier function $\mathrm{B}(b,x_i)$ in \eqref{eq:bar3nonconvex}. With the same reasoning in Section~\eqref{sec:datadriven}, solving the above RP is not straightforward generally. Therefore, we construct an SP by taking samples and then connect the solution of the obtained scenario programming to the safety of the stochastic system in \eqref{eq:mainsystem}. \payam{By collecting} i.i.d. samples $x_i,\; i \in \{1,\ldots,N\}$, from an assigned probability distribution over the state set, and approximating the expectation term in \eqref{eq:bar3nonconvex} results in a non-convex programming as the following:
\begin{align} 
	\text{SP\textsubscript{$\scriptscriptstyle N,\hat{N}$}}\!:\!\left\{
	\begin{array}{ll}
		\underset{d}{\min}\quad \mathcal{K}\quad \quad\\
		\text{s.t.}\quad\;
		\max_z\,\bar g_z(x_i,d) \!\leq\! 0, \,\,\forall i\in \{1,\ldots,N\},\\
		\qquad\qquad ~\;\;z \!\in\!\{1,\ldots,4\},\\
		%\max\big(g_z(x_i,d),\bar{g}_5(x_i,d)\big) \!\!\leq\! 0,z \!\in\!\{1,\ldots,4\},\quad \\
		%\quad\quad~~ \forall x_i \in X, \forall i\in\{1,\ldots,N\},\\
		\qquad ~~d=[\mathcal{K};\lambda;c;b_{\iota_1,\ldots,\iota_n};\kappa],\\
		\qquad ~~\mathcal{K}\in \mathbb{R},\; \lambda > 1,\; c \geq 0,\; \kappa \in (0,1), 
	\end{array}
	\right.
	\label{eq:SCN2_non}
\end{align}
where $\bar g_z:=g_z$ for all $z\in\{1,2,3\}$ and
\begin{align}
	\bar{g}_4(x_i,d)=\frac{1}{\hat{N}}\sum_{j=1}^{\hat{N}}\mathrm{B}(b,f(x_i,w_j))-  \;\kappa\;\mathrm{B}(b,x_i)-c+\delta-\mathcal{K}.
	\label{eq:scp_non}
\end{align}
Note that in this new scenario programming, we eliminated the constraint that forces a fixed \payam{probability} lower bound $1-\rho$ on the safety of the stochastic system, namely, $g_4$ in \eqref{eq:sampledcondition}. Instead, we are interested in providing the tightest possible lower bound of the safety probability according to Remark~\ref{remark:non_deterministic_verification}.
The main issue underlying here is that by considering $\kappa \in(0,1)$, the obtained scenario program is not convex anymore, and accordingly, one cannot naively utilize the results proposed in Theorems~\ref{theo:peyman}. \payam{Hence}, one cannot solve the SP in \eqref{eq:SCN2_non} \payam{by simply applying bisection over} $\kappa,$ while still utilizing the proposed results in the previous sections. 
% \textcolor{red}{what kind of bound do you get when $\kappa<1$? You can't just claim something in the air. I suggest you should discuss mathematically what you get from having \eqref{eq:bar3nonconvex}.}

Now we state the main problem we aim to address in this section.
\begin{problem}\label{problem_nonconvex}
	Consider an unknown dt-SS $\mathcal{S}$ as in Definition~\ref{def:mainsys}. Compute the \payam{largest} lower bound $(1-\rho) \in [0,1]$ on the probability of satisfying $\Psi$, \emph{i.e.,}
	\begin{equation*}
		\pr_w\big(\mathcal{S}\models_{\mathcal{H}} \Psi\big)\ge 1-\rho,
	\end{equation*}
	according to \eqref{eq:new_bound} together with a confidence $(1-\beta) \in [0,1]$ using a dataset $\mathcal{D}$ of the form \eqref{eq:data}. Moreover, establish a connection between the required size of dataset $\mathcal{D}$, the cardinality of \payam{the} set from which the parameter $\kappa$ is selected, and the desired confidence $1-\beta$.
	%Note that parameter $\rho$ in this problem should be computed according to \eqref{eq:new_bound} and based on optimal values of SP\textsubscript{$\scriptscriptstyle N,\hat{N}$}.
\end{problem}

%\end{mdframed}\vspace{0.2cm}
%Similar to the procedure described in Section~\ref{sec:safety}, the safety problem of the stochastic system in Definition~\ref{def:barrier} and Theorem~\ref{theo:kushner} can be reformulated as a robust optimization problem like \eqref{eq:RCP} but with a modified constraint $g_5$ as:
%\begin{align}
%g_5(x,d)=\EE\Big[\mathrm{B}(b,f(x,w))\mid x\Big]-\kappa\;\mathrm{B}(b,x)-c-\mathcal{K}.
%\end{align}

%As it is observed, the overall obtained optimization problem is not an RCP anymore due to the multiplication of the new decision variable $\kappa$ and coefficients of $\mathrm{B}(b,x)$. After sampling data and approximating the empirical mean, we resort to a new scenario program like \eqref{eq:SCN2}, but with a new $\bar{g}_5$ as:
%\begin{align}
%\bar{g}_5(x_i,d)=\frac{1}{\hat{N}}\sum_{j=1}^{\hat{N}}\mathrm{B}(b,f(x_i,w_j))-  \;\kappa\;\mathrm{B}(b,x_i)-c+\delta-\mathcal{K},
%\label{eq:scp_non}
%\end{align}
%where $g_5(x,d)$ is assumed to be a continuous function with Lipschitz constant $\mathrm{L}_{x_5}$.

In the next theorem, we present our solution to Problem~\ref{problem_nonconvex} by proposing a new confidence bound which is always valid even for the non-convex scenario program in \eqref{eq:SCN2_non}.

%\begin{theorem}\label{Thm3}
%	Consider a non-convex scenario program in \eqref{eq:SCN2_non} with $\kappa \in(0,1)$. Let $\epsilon,\beta \in[0,1]$, and $\mathrm{M}$ be the cardinality of a
%	finite set that 
%	$\kappa$ takes value from it. Suppose that $\mathcal{K}^*(\mathcal{D})$ is the optimal objective value of the scenario programming, then for a given cardinality $\mathrm{M}$, if $\mathcal{K}^*(\mathcal{D})+\epsilon\leq0$ 
%	\begin{align}
%	\pr_w(\mathcal{S}\models \Psi)\geq 1-\rho^\ast,
%	\end{align}
%where $\rho ^\ast$ is computed from relation in \eqref{eq:new_bound} using optimal values of SP\textsubscript{$\scriptscriptstyle N,\hat{N}$} namely $c^ \ast$ and $\lambda^ \ast$. The above statement is valid with a confidence of at least $1-\mathrm{M}(\beta+\beta_s)$ provided that $N\geq N\big(\bar{\epsilon},\beta\big)$ with $N\big(\bar{\epsilon},\beta\big)$ defined in \eqref{eq:number_samples}, in which $\bar{\epsilon}:=(\frac{\epsilon}{\mathrm{L}_{x}})^n$ with $\mathrm{L}_{x}:=\operatorname{max}\big(\mathrm{L}_{x_1},\mathrm{L}_{x_2},\mathrm{L}_{x_3},\mathrm{L}_{x_4}\big)$.
%\end{theorem}
\payam{
	\begin{theorem}\label{Thm3}
		Consider an unknown dt-SS as in \eqref{eq:mainsystem} together with the safety specification $\Psi$. Let $\mathrm{M}$ be the cardinality of a
		finite set from \Ali{which 
			$\kappa$ takes value in (0,1). Suppose that Assumptions~\ref{assum:lip}-\ref{ass:G} hold for the RP in \eqref{eq:RCP_non} with function $G(\cdot)$ and $\mathrm{L}_{x}:=\operatorname{max}\big(\mathrm{L}_{x_1},\mathrm{L}_{x_2},\mathrm{L}_{x_3},\mathrm{L}_{x_4}\big)$, where $\mathrm{L}_{x_i}, i \in\{1,\ldots,4\},$ is an upper bound on the Lipschitz constant of the $i^{th}$ constraint in \eqref{eq:RCP_non}. 
			Assume $\hat{N}$ is selected for the SP\textsubscript{$\scriptscriptstyle N,\hat{N}$} similar to Theorem~\ref{theo:empirical} in order to provide confidence $1-\beta_s$. Suppose $\mathcal{K}^\ast(\mathcal{D})$ is the optimal value of the optimization problem in \eqref{eq:SCN2_non} using $\hat{N}$ and $N$. Furthermore, $\epsilon=\mathrm I^{-1}(1-\mathrm M\beta;\mathcal Q+3,N-\mathcal Q-2)$ for $\beta \in [0,1]$, where $\mathcal{Q}$ is the number of coefficients of the barrier certificate. 
			Then the following statement holds with a confidence of at least $1-3\beta-\beta_s$: if $\mathcal{K}^\ast(\mathcal{D})+\mathrm L_xG^{-1}(\epsilon) \leq 0$, then
			\begin{align}
				\pr_w(\mathcal{S}\models_{\mathcal{H}} \Psi)\geq 1-\rho^\ast,
			\end{align}
			%placeholder
			where $\rho ^\ast$ is computed as in \eqref{eq:new_bound} using optimal solutions of SP\textsubscript{$\scriptscriptstyle N,\hat{N}$}, namely, $c^ \ast$, $\lambda^ \ast$, and $\kappa^\ast$. More importantly, with a confidence of at least $1-3\beta-\beta_s$, $\mathrm B(b^\ast,x)$ is a barrier certificate for $S$, satisfying \eqref{eq:bar1}, \eqref{eq:bar2}, and \eqref{eq:bar3nonconvex}, where $b^\ast$ is the optimal solution of SP\textsubscript{$\scriptscriptstyle N,\hat{N}$}.}
	\end{theorem}
}
%\begin{proof}
%	The required number of samples $N(\bar{\epsilon},\beta)$ can be investigated like the reasoning in the proof of Theorem~\eqref{theo:peyman}. The lower bound on the confidence can be inferred similar to the reasoning in the proof of the aforementioned theorem by defining the events
%	$\mathcal A_z  := \big\{\mathcal D\,\;|\, \mathcal{K}_z^*\!\leq\mathcal{K}^\ast_{\mathsf m_z}(\mathcal D)+\epsilon\big\}$, $\mathcal B_z := \big\{\mathcal D\,|\,\mathcal{K}^\ast_{\mathsf m_z}(\mathcal D)\le \mathcal K_z^\ast(\mathcal D)\big\}$, $\mathcal C_z :=\big\{\mathcal D\,|\,\mathcal K_z^\ast(\mathcal D)+\epsilon\le 0,\big\}$, and $z\in\{1,\ldots,\mathrm{M}\}$, where $\mathcal{K}_z^\ast$, and $\mathcal{K}_{m_z}^\ast(\mathcal{D})$ are optimal objective values regarding the original RP, and its equivalent SP before empirical approximation of the expectation term in $g_4$ for different values of parameter $\kappa$ from the cardinality set $z\in\{1,\ldots,\mathrm{M}\}$. 
%According to the previous proof, one has:
%\begin{align}
%\pr(\mathcal A_z \cap \mathcal B_z)\geq 1-\pr(\mathcal A_z^c)-\pr(\mathcal B_z^c)\nonumber
%\geq 1-\beta-\beta_s.
%\end{align}
%\textcolor{red}{To be continued ...} 
%Maybe by defining a new event like $\mathcal{T}_z:=\mathcal A_z \cap \mathcal B_z$, we can deduce that $\pr\Big(\mathcal{T}_1 \cap \ldots\cap\mathcal{T}_M\Big)\geq 1- \mathrm{M}\;(\beta+\beta_s)$ as we are interested in intersection of these events?
\begin{proof}
	%placeholder
	%\textcolor{orange}{
	Denote the optimal values of the RP and its equivalent scenario programming before the empirical approximation of the expectation term in $g_4$, namely, SP\textsubscript{$\scriptscriptstyle N$}, by $\mathcal K^\ast$ and $\mathcal K^\ast_{\mathsf m}(\mathcal D)$, respectively. \Ali{Similar to \eqref{eq:proof1}, one has
		\begin{align*}
			\pr\big(\mathcal{K}^*\leq\mathcal{K}^\ast_{\mathsf m}(\mathcal D)+\mathrm{L}_xG^{-1}(\epsilon)\big)\geq 1-3\beta,
		\end{align*}
		for any $N\geq \tilde{N}\big(\epsilon_1,\ldots,\epsilon_{\mathrm{M}},\beta\big)$,
		where
		\begin{align*}
			&\tilde{N}\big(\epsilon_1,\ldots,\epsilon_{\mathrm{M}},\beta\big):=\nonumber\\&  \min \Big\{N\in\mathbb{N}\mid\sum_{z=1}^{\mathrm{M}}\sum_{i=0}^{\mathrm{d}-1} \dbinom{N}{i}\epsilon_z^{\;i}
			(1-\epsilon_z)^{N-i} \leq \beta \Big\}.
		\end{align*}
		Alternatively, one can set $\epsilon:=\epsilon_1 = \epsilon_2 = \ldots = \epsilon_{\mathrm{M}}$ in the above expression to get the inequality $\epsilon \leq \mathrm I^{-1}(1-\mathrm M\beta;\mathrm d,N-\mathrm d +1)$, where $\mathrm{M}$ is the cardinality of the set \payam{from which} $\kappa$ is selected, and $\mathrm{d}$ is the number of decision variables. By choosing $\mathrm{d}:=\mathcal{Q}+3$, one gets the parameters of the incomplete beta function in the theorem statement.
		On the other hand, due to the particular selection of $\hat N$ and $\beta_s$ similar to Theorem~\ref{theo:empirical}, it can be deduced that
		\begin{equation*}
			\label{eq:feasibility}
			\pr_w\Big(\hat{\mathrm{B}}(b,x\,|\,\mathcal{D}) \models\text{SP\textsubscript{$\scriptscriptstyle N$}}\Big)\geq 1- \beta_s,
		\end{equation*}
		where $\hat{\mathrm{B}}(b,x\,|\,\mathcal{D})$ is the barrier function whose coefficients are the optimal solution of SP\textsubscript{$\scriptscriptstyle N$}. Therefore, we have
		\begin{equation}
			\label{eq:optimal_compare_u}
			\pr\left(\mathcal{K}^\ast_{\mathsf m}(\mathcal D)\le \mathcal K^\ast(\mathcal D)\right)\ge 1-\beta_s.
		\end{equation}
		By defining events
		$\mathcal A  := \{\mathcal D\,|\, \mathcal{K}^*\!\leq\mathcal{K}^\ast_{\mathsf m}(\mathcal D)+\mathrm L_xG^{-1}(\epsilon)\}$,
		$\mathcal B := \{\mathcal D\,|\,\mathcal{K}^\ast_{\mathsf m}(\mathcal D)\le \mathcal K^\ast(\mathcal D)\}$,
		and $\mathcal C :=\{\mathcal D\,|\,\mathcal K^\ast(\mathcal D)+\mathrm L_xG^{-1}(\epsilon)\le 0\}$,
		where $\pr(\mathcal A)\ge 1-3\beta $ and  $\pr(\mathcal B)\ge 1-\beta_s$, it is easy to conclude using the same reasoning as in the second part of proof of Theorem~\eqref{theo:peyman} that
		\begin{align*}
			\pr(\mathcal{K}\leq 0)\geq 1-3\beta-\beta_s,
		\end{align*}
		which ensures safety of the stochastic system with a lower bound $1-\rho$ and a confidence of at least $1-3\beta-\beta_s$.
		%}
	}
\end{proof}

\section{Numerical Examples}
\label{sec:case_study}

\Ali{The simulations of this section are performed on an iMac 3.5 GHz Quad-Core Intel Core i7. The optimizations are solved by CVX Toolbox \cite{cvx} with Mosek \cite{andersen2000mosek} as the solver.}

\subsection{Temperature verification for three rooms}
Consider a temperature regulation problem for three rooms characterized by the following discrete-time stochastic system:
\begin{align}
	\mathrm{T}_1(t+1)=&\big(1-\tau_s(\alpha+\alpha_e)\big)\mathrm{T}_1(t)+\tau_s\alpha\mathrm{T}_2(t)+\nonumber \\&\tau_s\alpha_eT_e+w_1(t)\nonumber \\
	\mathrm{T}_2(t+1)=&\big(1-\tau_s(2\alpha+\alpha_e)\big)\mathrm{T}_2(t)+\tau_s\alpha(\mathrm{T}_1(t)+\mathrm{T}_3(t))+ \nonumber\\&\tau_s\alpha_eT_e+w_2(t) \nonumber \\
	\mathrm{T}_3(t+1)=&\big(1-\tau_s(\alpha+\alpha_e)\big)\mathrm{T}_3(t)+\tau_s\alpha\mathrm{T}_2(t)+\nonumber \\&\tau_s\alpha_eT_e+w_3(t),
	\label{eq:tem_3}
\end{align}
where $\mathrm{T}_1(t)$, $\mathrm{T}_2(t)$, and $\mathrm{T}_3(t)$ are temperatures of three rooms, respectively. \payam{Terms} $w_1(t)$, $w_2(t)$, and $w_3(t)$ are additive zero-mean Gaussian noises with standard deviations of $0.01$, which model the environmental uncertainties. Parameter $T_e = 10 \degree C$ is the ambient temperature. Constants $\alpha_e = 8\times 10^{-3}$ and $\alpha=6.2\times10^{-3}$ are heat exchange coefficients between rooms and the ambient, and individual rooms, respectively. The model for each room is adapted from \cite{girard2016safety} discretized by $\tau_s=5$ minutes. Let us consider the regions of interest for each room as $X_{in} = [17\degree C,18\degree C]$, $X_{u} = [29\degree C,30\degree C]$, and $X = [17\degree C,30\degree C]$. We assume the model of the system and the distribution of the noise are unknown. 
%\begin{figure}[h]
%	\centering 
%	\includegraphics[width=.6
%	\linewidth]{3_rooms.eps}
%	\caption{The cascaded arrangement of three rooms whose related dt-SS is introduced in \eqref{eq:tem_3}. \Sadegh{Eliminate this figure. It doesn't say much, only three rooms sitting in a row.}}
%	\label{fig:arrangement}
%\end{figure}
The main goal is to verify whether the temperature of each room remains in the comfort zone $[17,29]$ for the time horizon $\mathcal{H}=3$ which is equivalent to $15$ minutes, with a priori confidence of $99\%$.

Let us consider a barrier certificate with degree $k=2$ in the polynomial form as $[T1;T2;T3]^T\mathrm{P}[T1;T2;T3]=b_0T_1^2+b_1T_2^2+b_2T_3^2+b_3T_1T_2+b_4T_1T_3+b_5T_2T_3+b_6T_1+b_7T_2+b_8T_3+b_9$,
where
\begin{align}
	\mathrm{P}=\begin{bmatrix}
		b_0 & \frac{b_3}{2} & \frac{b_4}{2}&\frac{b_6}{2}\\
		\frac{b_3}{2} & b_1 & \frac{b_5}{2}&\frac{b_7}{2}\\
		\frac{b_4}{2}&\frac{b_5}{2}&b_2&\frac{b_8}{2}\\
		\frac{b_6}{2}&\frac{b_7}{2}&\frac{b_8}{2}&b_9
	\end{bmatrix}.
	\label{eq:firstpmatrix}
\end{align}
According to Algorithm \ref{alg:verification}, we first choose the desired confidence parameters $\beta$ and $\beta_s$ as $\frac{0.005}{3}$ and $0.005$, respectively. The value of empirical approximation error is selected as $\delta=0.05$. We choose $\rho=0.2$. \Ali{The Lipschitz constant is computed as $1.5$ according to Remark~\ref{rem:data_lip}.
	By enforcing $\hat{M}=0.005$, the required number of samples for the approximation of the expected value in \eqref{eq:SCN2} is $\hat{N}=400$. Now, we solve the scenario problem SCP\textsubscript{$\scriptscriptstyle N,\hat{N}$} with the number of samples $N=6 \times 10^6$ and the computed $\hat{N}=400$, which gives us the optimal objective value $\mathcal{K}^*(\mathcal{D}) = -0.46$. The computation time is about $5$ minutes. For $N=6\times 10^6$ and $\beta=\frac{0.005}{3}$, $\epsilon $ is computed as $4.36\times 10^{-6}$. Function $G^{-1}(\epsilon)$ is also computed as $16.09\epsilon^{\frac{1}{3}}$ according to Corollary \ref{rem:gofepsilon}.} 

\Ali{Since $\mathcal{K}^*(\mathcal{D})+\mathrm{L}_xG^{-1}(\epsilon) =-0.066\leq0$, according to Theorem~\ref{theo:peyman}, one can conclude:
	\begin{align*}
		\pr_w(\mathcal{S}\models_{3} \Psi)\geq 1-\rho=0.80, 
	\end{align*}
	with a confidence of at least $1-3\beta-\beta_{s}=0.99$. The barrier certificate constructed from solving SCP\textsubscript{$\scriptscriptstyle N,\hat{N}$} is as follows:
	\begin{align}
		\hat{\mathrm{B}}(b,T_1,T_2,T_3\mid&\mathcal{D})=0.112T_1^2+0.112T_2^2
		+0.112T_3^2\nonumber\\&-0.004T_1T_2-0.005T_1T_3-0.002T_2T_3\nonumber\\&-3.761T_1-3.815T_2-3.803T_3+99.93.
	\end{align}
}
\Ali{The computed optimal values for $c$ and $\lambda$ are $0.627$ and $14.872$, respectively.} The scatter plot of the obtained barrier certificate is illustrated in Fig.~\ref{fig:scatter}. As can be seen in this figure, the barrier certificate has less values in the initial set while it has larger values in the unsafe region.

\Ali{
	We remark that the conservatism of our approach is originating from two sources. (a) The first one is that we are using barrier certificates for computing the lower bound. A barrier certificate with a fixed template (polynomial of a certain degree) gives a lower bound that could have a gap with the best lower bound on the safety probability. (b) Our sampling approach requires making the optimization more conservative to account for going from robust programs over continuous (uncountable) domains to a scenario program with finite number of samples.  
	If one assumes that the model is known in this case study, the synthesized barrier certificate has the parameters $c = 0.9767$ and $\lambda = 31.51$. This gives the lower bound $0.875$ on the safety probability.
	Therefore, our approach provides a more conservative lower bound $0.80$ since it assumes no knowledge of the model.
}
\begin{figure}[ht]
	\centering 
	\includegraphics[width=.7
	\linewidth]{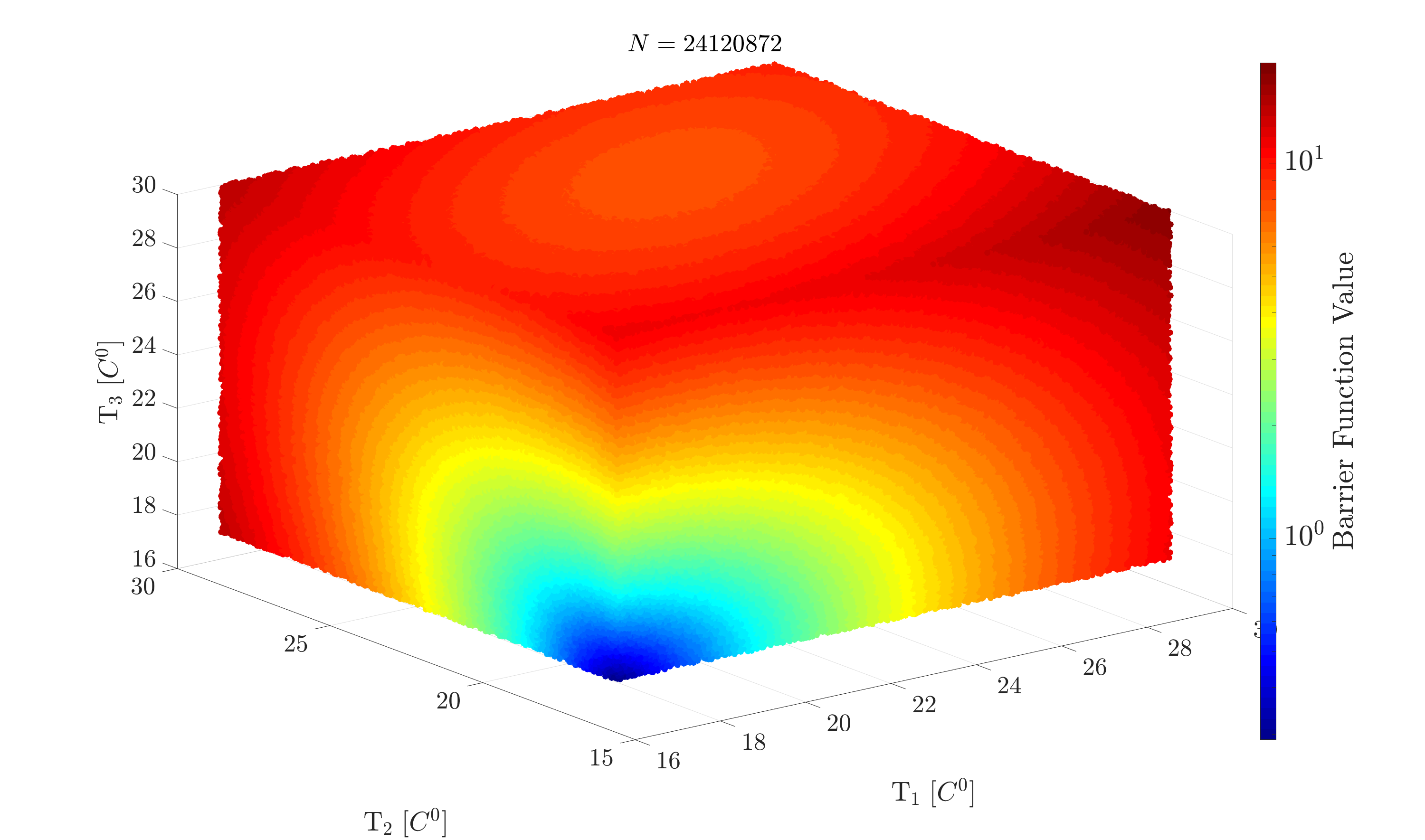}
	\caption{\payam{Scatter plotting of the barrier certificate indicating portions of the state set where the inequalities in \eqref{eq:SCN2} are enforced for $6\times 10^6$ sampled data.}}
	\label{fig:scatter}
\end{figure}
\subsection{Lane keeping system}
Lane keeping assist system is a future development of the modern lane departure warning system embedded in the current vehicles. This system usually assists the driver through electronic assistance with the steering force. The characteristics of this support depends on the distance of the vehicle from the edge of the lane among other factors such as uncertainties\cite{Anu:2013}. One of the key challenges in such assisting systems is verifying the obtained performance which can be defined as a safety problem. 

In this subsection, it is supposed that the model of the vehicle and the distribution of noise are unknown, and one only has access to a finite number of samples. This unknown system is characterized by a simplified kinematic single-track model of BMW320i which is adapted from \cite{althoff2017commonroad} by discretization of the model and adding noise to imitate the uncertainties.

The nonlinear stochastic difference equation is as follows:
\begin{align}
	&x(t+1) = x(t) + \tau_sv\;\cos(\psi(t)+\mathrm{b}) + w_1(t)\nonumber \\
	\mathcal{S}:\;
	& y(t+1) = y(t) + \tau_sv\;\sin(\psi(t)+\mathrm{b}) + w_2(t)\nonumber \\
	& \psi(t+1) = \psi(t) + \frac{ \tau_s v}{l_{r}}\sin(\mathrm{b})+w_3(t),
\end{align}
where $\mathrm{b}=\frac{l_r}{l_r+l_f}\tan^{-1}(\delta_f)$ with $\delta_f=5$ degrees as the steering angle. \payam{Parameters} $l_r=1.384$ and $l_f=1.384$ are the distances between the center of gravity of the vehicle to the rear and front axles, respectively. Variables $x$, $y$, and $\psi$ denote horizontal movement, vertical movement, and the heading angle, respectively. This system is considered to be affected by zero-mean additive noises $w_1$, $w_2$, and $w_3$ which are related to uncertainties of position $x$, position $y$, and the heading angle $\psi$ with standard deviation of $0.01$, $0.01$, and $0.001$ respectively. Other parameters are the sampling time $(\tau_s = 0.1s)$, and the velocity $(v=5m/s)$.

The state set is considered as $X = [1,10] \times [-7,7] \times [-0.05,0.05]$. The regions of interest are $X_{in} =[1,2]\times [-0.5,0.5] \times [-0.005,0.005]$, $X_{u_1} = [1,10] \times [-7,-6]\times [-0.05,0.05] $, and $X_{u_2} = [1,10] \times [6,7]\times [-0.05,0.05]$. 
Now, the goal is to verify if the vehicle does not enter the unsafe regions of the lane for the time horizon of $\mathcal{H}=3$ or equivalently $0.3\;s$ with a desired confidence of $90\%$.

We consider a barrier certificate of degree $k=2$ in the polynomial form as
$[x;y;\psi]^T\mathrm{P}[x;y;\psi]=b_0x^2+b_1y^2+b_2\psi^2+b_3xy+b_4x\psi+b_5y\psi+b_6x+b_7y+b_8\psi+b_9$, where the matrix $\mathrm{P}$ is \payam{as in \eqref{eq:firstpmatrix}}.

\Ali{We follow Algorithm \ref{alg:verification} to find the barrier certificate and providing a probabilistic guarantee on the safety of stochastic system. First, the desired confidence parameters $\beta$ and $\beta_s$ are chosen as $\frac{.095}{3}$ and $0.005$, respectively. We also select the empirical approximation error $\delta=0.02$. The desired lower bound of safety probability is selected as $1-\rho=0.80$. The Lipschitz constant is computed as $\mathrm{L}_x=10$ according to Remark \ref{rem:data_lip}. 
	By enforcing $\hat{M}=0.006$, the required number of samples for the approximation of the expected value in \eqref{eq:SCN2} is $\hat{N}=3000$. Now, we solve the scenario problem SCP\textsubscript{$\scriptscriptstyle N,\hat{N}$} with an arbitrary sample number $N=6\times 10^6$ and $\hat{N}$ which gives us the optimal value $\mathcal{K}^*(\mathcal{D}) = -0.4518$. The computation time is about $5$ minutes. For those values of samples $N$ and $\beta$, $\epsilon$ is computed as $3.41\times 10^{-6}$. Using Corollary~\ref{rem:gofepsilon}, $G^{-1}(\epsilon)$ is computed as $2.92\epsilon^{\frac{1}{3}}$.} 

\Ali{Since $\mathcal{K}^*(\mathcal{D})+2.92\;\mathrm{L}_x\epsilon^{\frac{1}{3}}=-0.01\leq0$, according to Theorem~\ref{theo:peyman}, one can deduce that}
\begin{align*}
	\pr_w(\mathcal{S}\models_{3} \Psi)\geq 1-\rho=0.80, 
\end{align*}
\Ali{with a confidence of at least $1-3\beta-\beta_{s}=90\%$. The barrier certificate constructed from solving SCP\textsubscript{$\scriptscriptstyle N,\hat{N}$} is represented as:
	\begin{align}
		\hat{\mathrm{B}}(b,x,y,\psi\mid\mathcal{D})=&0.39y^2+0.15\psi^2
		+0.009x\psi\nonumber\\&-0.007y\psi-0.015\psi+0.452.
\end{align}}
\begin{figure}[t]
	\centering 
	\includegraphics[width=.7
	\linewidth]{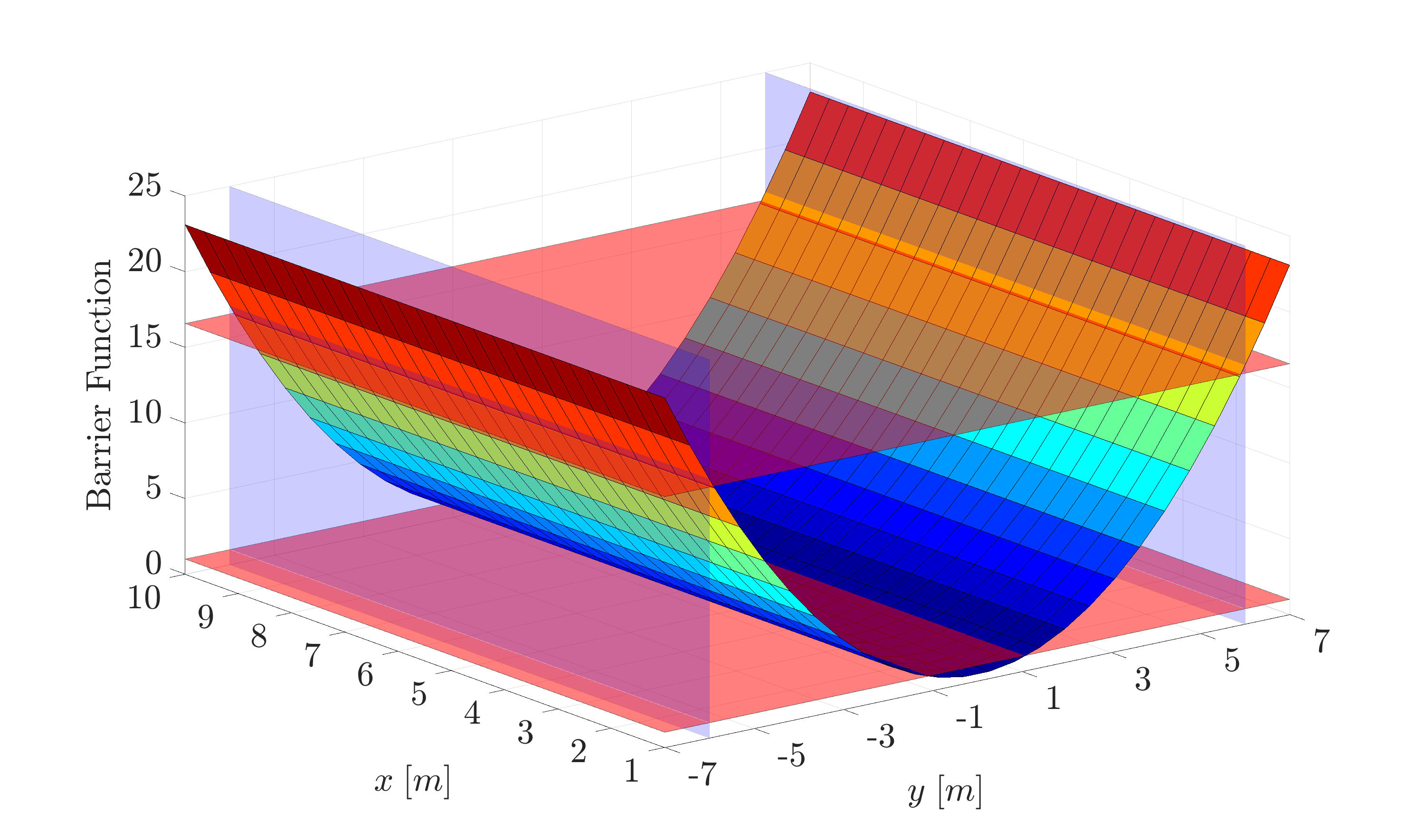}
	\caption{Surface plot of the barrier certificate $\mathrm{B}(x,y,\psi)$ with respect to $x$ and $y$ for fixed $\psi=0$.}
	\label{fig:vehicle_1}
\end{figure}
\begin{figure}[b]
	\centering 
	\includegraphics[width=.7
	\linewidth]{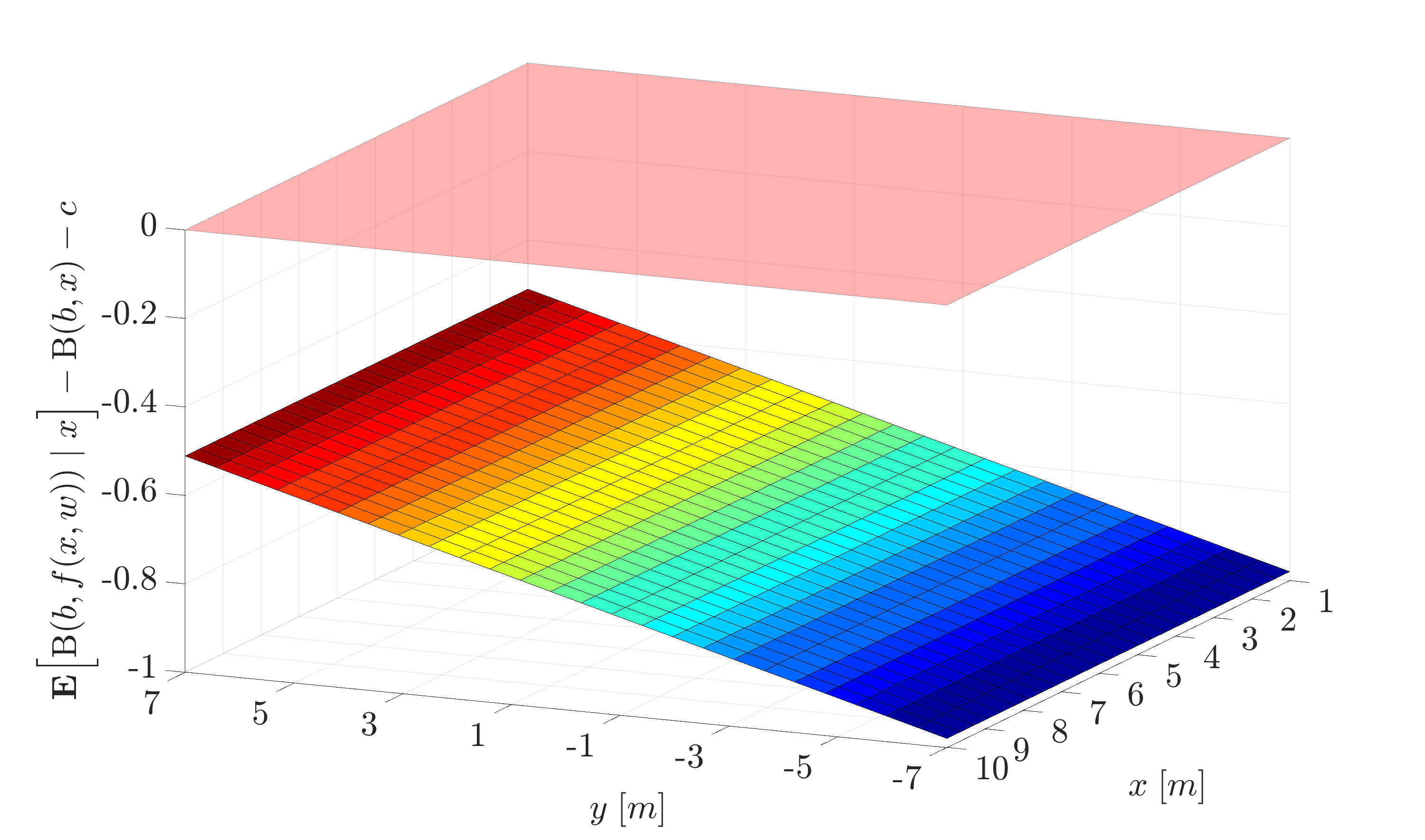}
	\caption{Satisfaction of the third condition in Definition~\ref{def:barrier} (for $\psi=0$) $\mathrm{B}(x,y,\psi)$ based on collected data.}
	\label{fig:vehicle_2}
\end{figure}
\Ali{The optimal values of $c$ and $\lambda$ are $0.57$ and $14.04$, respectively. The exact value of the coefficients are reported in the appendix.} 

%The computation time lasts about $2$ days using an iMac 4.2 GHz Quad-Core Intel Core i7, 32 GB RAM 2400 MHz DDR4 under macOS Big Sur version 11.1. The large value of execution time is because of large number of data which is due to the nature of the proposed approach which aims at preventing the worst cases to be happened. However, it should be noticed that, we encounter a linear programming and massive improvements such as large-scaled techniques are under progress. Furthermore, developments in the hardware setting like high RAM capacities can significantly improve the computation time and make the proposed approach workable even for higher dimensions and more challenging problems. 

The surface plot of the barrier certificate $\mathrm{B}(x,y,\psi)=\hat{\mathrm{B}}(b,x,y,\psi\mid\mathcal{D})$ with respect to $x$ and $y$ for a fixed value of $\psi=0$ is depicted in Fig.~\ref{fig:vehicle_1}. The blue transparent planes separate unsafe region on $y$, while the lower and upper red transparent planes demonstrate the thresholds in constraints \eqref{eq:bar1} and \eqref{eq:bar2}, respectively. Satisfaction of the first and second condition of barrier certificate in Definition~\ref{def:barrier} can be observed in Fig.~\ref{fig:vehicle_1}. The satisfaction of the third condition is illustrated in Fig.~\ref{fig:vehicle_2}.

\subsection{Synthesizing a temperature controller}
Consider a temperature regulation problem for a room using a heater characterized by 
\begin{align}
	\Ali{\mathcal{S}}:\; \mathrm{T}(t+1)=&\mathrm{T}(t) + \tau_s\big(\alpha_e(T_e-\mathrm{T}(t))+\nonumber\\&\alpha_h(T_h-\mathrm{T})u(t)\big)+w(t),
	\label{eq:tem_1}
\end{align}
where $w(t)$ is a zero-mean Gaussian noise with standard deviation of $0.05$. Parameters are $T_e= 15$, $T_h= 45$, $\alpha_e=8\times 10^{-3}$, $\alpha_h=3.6\times 10^{-3}$, and $\tau_s=5$. Regions of interest are defined as $X_{in}=[22\degree C,23\degree C]$, $X_{u_1} = [27\degree C,28\degree C]$, $X_{u_2} = [16.5\degree C,17.5\degree C]$, and $X = [16.5\degree C,28\degree C]$. The input region is $[0,1]$.
We assume that the model of the system and the distribution of the noise are unknown. 
The main goal is to design a controller that forces the temperature to remain in the comfort zone $[17.5,27]$ for the time horizon $\mathcal{H}=60$, which is equivalent to $300$ minutes, with a priori confidence of $95\%$.

Let us fix a control barrier certificate with degree $k=4$ in the polynomial form as $T^T\mathrm{P}T=b_0T^4+b_1T^3+b_2T^2+b_3T+b_4$ with $b_0, b_1,b_2,b_3,b_4 \in \mathbb{R}$. The structure of the controller is considered to be a polynomial of degree $k'=4$ as $u(p^1,T)=T^T\mathrm{P}_uT=p_0T^4+p_1T^3+p_2T^2+p_3T+p_4$. Matrices $\mathrm{P}$ and $\mathrm{P}_u$ can be represented as:
\begin{align}
	\mathrm{P}=\begin{bmatrix}
		b_0 & \frac{b_1}{2} & \frac{b_2}{3}\\
		\frac{b_1}{2} &\frac{b_2}{3}& \frac{b_3}{2}\\
		\frac{b_2}{3}&\frac{b_3}{2}&b_4
	\end{bmatrix},
	\mathrm{P}_u=\begin{bmatrix}
		p_0 & \frac{p_1}{2} & \frac{p_2}{3}\\
		\frac{p_1}{2} &\frac{p_2}{3}& \frac{p_3}{2}\\
		\frac{p_2}{3}&\frac{p_3}{2}&p_4
	\end{bmatrix}.
\end{align}
According to Algorithm \ref{alg:synthesiss}, we first choose the desired confidences $\beta$ and $\beta_s$ as $\frac{0.005}{3}$ and $0.045$ respectively. We also select the approximation error $\delta=2$. \Ali{The Lipschitz constant $\mathrm{L}_{x,u}$ is computed as $12$ according to Remark~\ref{rem:data_lip}.
	By considering $\hat{M}=1.5\times10^5$, the required number of samples for the approximation of the expected value in \eqref{eq:SCN2} is $\hat{N}=833330$. Now, we solve the scenario problem SCP\textsubscript{$\scriptscriptstyle N,\hat{N}$} with the selected number of samples $N=1.5\times 10^6$ and $\hat{N}$ which gives us the optimal value $\mathcal{K}^*(\mathcal{D}) = -0.41$. The computation time is about $2$ minutes. For $N=1.5 \times 10^6$ and $\beta=\frac{0.005}{3}$, value of $\epsilon$ is computed as $1.7424\times 10^{-5}$. Using Corollary~\ref{rem:gofepsiloninput}, $G^{-1}(\epsilon)$ is computed as $4.91\epsilon^{\frac{1}{2}}$.}

\Ali{Since $\mathcal{K}^*(\mathcal{D})+\mathrm{L}_{x,u}G^{-1}(\epsilon) = -0.164 \leq 0$, one has
	\begin{align*}
		\pr_w^{\mathrm{p}}(\mathcal{S}\models_{60} \Psi)\geq 1-\rho=0.80, 
	\end{align*}
	with a confidence of at least $1-3\beta-\beta_{s}=95\%$. The computed values for $\lambda$ and $c$ are $4817$ and $16.04$, respectively. The control barrier certificate constructed from solving SCP\textsubscript{$\scriptscriptstyle N,\hat{N}$} is:
	\begin{align*}
		\hat{\mathrm{B}}(b,T\mid\mathcal{D}) =& \;11.89\;T^4-1.07\times 10^3\;T^3+3.61\times 10^{4}\;T^2\\&-5.42 \times 10^{5}+3.05 \times 10^6.
	\end{align*}
	The obtained controller is:
	\begin{align*}
		\mathscr{P}_1(p^1,T\mid\mathcal{D}) =& \;1.45\times 10^{-5}T^3+0.012T^2+0.355.
	\end{align*}
	The temperature trajectories for $15$ different realizations of noise from three different initial temperature in the range $[22\degree,23\degree]$ is illustrated in Fig.~\ref{fig:controllertmp}. 
	As can be seen, the temperature in the collected trajectories do not enter the unsafe set, which is in gray color. We also ran the system to get $10^4$ trajectories, all of them remain safe. This confirms the theoretical lower bound computed by our approach. 
}
\begin{figure}[ht]
	\centering 
	\includegraphics[width=.7
	\linewidth]{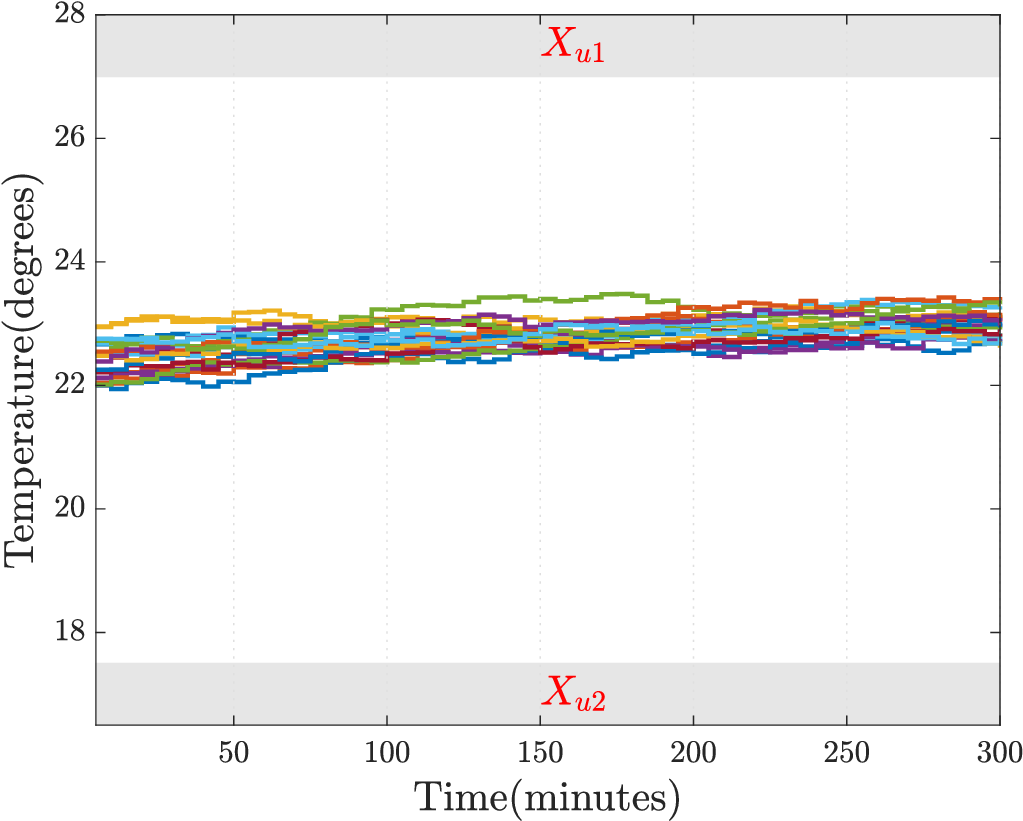}
	\caption{The temperature trajectories of $15$ different realizations of noise for three different initial temperature in the range $[22\degree,23\degree]$.}
	\label{fig:controllertmp}
\end{figure}

The conservativeness of our approach in terms of the safety bound $(1-\rho)$ and the number of samples is shown in Table~\ref{table:conservative}. 
The values are reported for increasing number of samples and two safety thresholds with $\rho \in\{0.1,0.2\}$.
As can be seen from the table, increasing the number of samples makes $\epsilon$ smaller and reduces the term $\mathrm{L}_{x,u}G^{-1}(\epsilon)$ used in \eqref{eq:main_condition_2}. In contrast, the values of $\mathcal{K}^*(\mathcal{D})$ become larger. This creates a tradeoff between the two terms in \eqref{eq:main_condition_2}. Note that the condition of having a negative value for $\mathrm{L}_{x,u}G^{-1}(\epsilon) + \mathcal{K}^*(\mathcal{D})$, thus guaranteeing safety with probability $(1-\rho)$, is only satisfied in the last two row of the table for $\rho = 0.2$ (indicated in blue color). Also, notice that the satisfaction of \eqref{eq:main_condition_2} for a higher desired safety probability requires larger number of samples.

\begin{table*}
		\centering
	\caption{Conservativeness of the proposed approach}
	\label{table:conservative}
	%\resizebox{\columnwidth}{!}{%
	\begin{tabular}{||c c c c c c c||} 
		\hline
		Number of samples  & Computed $\epsilon$  & $\mathrm{L}_{x,u}G^{-1}(\epsilon)$ & $\mathcal{K}_{\rho=0.2}^*(\mathcal{D})$ & $\mathcal{K}_{\rho=0.1}^*(\mathcal{D})$ & $\mathcal{K}_{\rho=0.2}^*(\mathcal{D})+\mathrm{L}_{x,u}G^{-1}(\epsilon)$ & $\mathcal{K}_{\rho=0.1}^*(\mathcal{D})+\mathrm{L}_{x,u}G^{-1}(\epsilon)$ \\ [0.5ex] 
		\hline\hline
		$10^3$ & $0.026$   & $9.5$ & $-0.48$ & $-0.45$ & \textcolor{black}{$9.02$} & \textcolor{black}{$9.05$}\\ 
		\hline
		$10^4$ & $0.003$  & $3$ & $-0.42$ & $0.09$ & \textcolor{black}{$2.58$} & \textcolor{black}{$3.09$}\\ 
		\hline
		$10^5$ & $2.61 \times 10^{-4}$  & $0.952$ & $-0.43$ &$1.37$ & \textcolor{black}{$0.522$} & \textcolor{black}{$2.32$}\\ 
		\hline
		$1.5 \times 10^6$ & $1.74 \times 10^{-5}$ & $0.246$ & $-0.41$ &$ 2.08$ & \textcolor{blue}{$-0.164$} & \textcolor{black}{$2.33$}\\ 
		\hline
		$3 \times 10^6$ & $8.71 \times 10^{-6}$  & $0.174$ & $-0.35$ & $2.09$ & \textcolor{blue}{$-0.176$} & \textcolor{black}{$2.26$}\\ [1ex] 
		\hline
	\end{tabular}
	%}
\end{table*}

\section{Conclusion}
\label{sec:ref}
\payam{
	We proposed a formal verification and synthesis procedure for discrete-time continuous-space stochastic systems with unknown dynamics against safety specifications. Our approach is based on the notion of barrier certificate and uses sampled trajectories of the unknown system.
	We first casted the computation of the barrier certificate as a robust convex program (RCP) and approximated its solution with a scenario convex program (SCP) by replacing the unknown dynamics with the sampled trajectories. We then established that the optimal solution of the SCP gives a feasible solution for the RCP with a given confidence, and formulated a lower bound on the required number of samples. Our approach provided a lower bound on the safety probability of the stochastic unknown system when the number of sampled data is larger than a specific lower bound that depends on the desired confidence. We extended the results to a class of non-convex barrier-based safety problems and showed the applicability of our proposed approach using three case studies.
}

\bibliographystyle{plain}
\bibliography{sample}

\begin{thebibliography}{10}

\bibitem{abate2020formal}
Alessandro Abate, Daniele Ahmed, Mirco Giacobbe, and Andrea Peruffo.
\newblock Formal synthesis of lyapunov neural networks.
\newblock {\em IEEE Control Systems Letters}, 5(3):773--778, 2020.

\bibitem{althoff2017commonroad}
Matthias Althoff, Markus Koschi, and Stefanie Manzinger.
\newblock Commonroad: Composable benchmarks for motion planning on roads.
\newblock In {\em 2017 IEEE Intelligent Vehicles Symposium (IV)}, pages
  719--726. IEEE, 2017.

\bibitem{andersen2000mosek}
Erling~D Andersen and Knud~D Andersen.
\newblock The mosek interior point optimizer for linear programming: an
  implementation of the homogeneous algorithm.
\newblock In {\em High performance optimization}, pages 197--232. Springer,
  2000.

\bibitem{BK08}
Christel Baier and Joost-Pieter Katoen.
\newblock {\em Principles of model checking}.
\newblock MIT press, 2008.

\bibitem{belta2017formal}
Calin Belta, Boyan Yordanov, and Ebru~Aydin Gol.
\newblock {\em Formal methods for discrete-time dynamical systems}, volume~15.
\newblock Springer, 2017.

\bibitem{berberich2020data}
Julian Berberich, Johannes K{\"o}hler, Matthias~A Muller, and Frank Allgower.
\newblock Data-driven model predictive control with stability and robustness
  guarantees.
\newblock {\em IEEE Transactions on Automatic Control}, 2020.

\bibitem{borrmann2015control}
Urs Borrmann, Li~Wang, Aaron~D Ames, and Magnus Egerstedt.
\newblock Control barrier certificates for safe swarm behavior.
\newblock {\em IFAC-PapersOnLine}, 48(27):68--73, 2015.

\bibitem{calafiore2006scenario}
Giuseppe~C Calafiore and Marco~C Campi.
\newblock The scenario approach to robust control design.
\newblock {\em IEEE Transactions on automatic control}, 51(5):742--753, 2006.

\bibitem{clark2021control}
Andrew Clark.
\newblock Control barrier functions for stochastic systems.
\newblock {\em Automatica}, 130:109688, 2021.

\bibitem{coulson2020distributionally}
Jeremy Coulson, John Lygeros, and Florian D{\"o}rfler.
\newblock Distributionally robust chance constrained data-enabled predictive
  control.
\newblock {\em arXiv:2006.01702}, 2020.

\bibitem{dawson2022safe}
Charles Dawson, Zengyi Qin, Sicun Gao, and Chuchu Fan.
\newblock Safe nonlinear control using robust neural lyapunov-barrier
  functions.
\newblock In {\em Conference on Robot Learning}, pages 1724--1735. PMLR, 2022.

\bibitem{Anu:2013}
Verband der Automobilindustrie.
\newblock Lane keeping assist systems.
\newblock
  \url{https://www.vda.de/en/topics/safety-and-standards/lkas/lane-keeping-assist-systems.html},
  2020.

\bibitem{esfahani2014performance}
Peyman~Mohajerin Esfahani, Tobias Sutter, and John Lygeros.
\newblock Performance bounds for the scenario approach and an extension to a
  class of non-convex programs.
\newblock {\em IEEE Transactions on Automatic Control}, 60(1):46--58, 2014.

\bibitem{girard2005reachability}
Antoine Girard.
\newblock Reachability of uncertain linear systems using zonotopes.
\newblock In {\em International Workshop on Hybrid Systems: Computation and
  Control}, pages 291--305. Springer, 2005.

\bibitem{girard2016safety}
Antoine Girard, Gregor G{\"o}ssler, and Sebti Mouelhi.
\newblock Safety controller synthesis for incrementally stable switched systems
  using multiscale symbolic models.
\newblock {\em IEEE {T}ransactions on {A}utomatic {C}ontrol, vol. 61, no. 6,
  pp. 1537--1549}, 2016.

\bibitem{cvx}
Michael Grant and Stephen Boyd.
\newblock {CVX}: Matlab software for disciplined convex programming, version
  2.1.
\newblock \url{http://cvxr.com/cvx}, March 2014.

\bibitem{han2015sublinear}
Shuo Han, Ufuk Topcu, and George~J Pappas.
\newblock A sublinear algorithm for barrier-certificate-based data-driven model
  validation of dynamical systems.
\newblock In {\em 54th IEEE conference on decision and control (CDC)}, pages
  2049--2054, 2015.

\bibitem{hernandez2001chebyshev}
MA~Hern{\'a}ndez.
\newblock Chebyshev's approximation algorithms and applications.
\newblock {\em Computers \& Mathematics with Applications}, 41(3-4):433--445,
  2001.

\bibitem{jagtap20202020control}
Pushpak Jagtap, George~J Pappas, and Majid Zamani.
\newblock Control barrier functions for unknown nonlinear systems using
  {G}aussian processes.
\newblock {\em arXiv:2010.05818}, 2020.

\bibitem{jagtap2019formal}
Pushpak Jagtap, Sadegh Soudjani, and Majid Zamani.
\newblock Formal synthesis of stochastic systems via control barrier
  certificates.
\newblock {\em IEEE Transactions on Automatic Control}, 66(7):3097--3110, 2020.

\bibitem{kanamori2012worst}
Takafumi Kanamori and Akiko Takeda.
\newblock Worst-case violation of sampled convex programs for optimization with
  uncertainty.
\newblock {\em Journal of Optimization Theory and Applications},
  152(1):171--197, 2012.

\bibitem{kenanian2019data}
Joris Kenanian, Ayca Balkan, Raphael~M Jungers, and Paulo Tabuada.
\newblock Data driven stability analysis of black-box switched linear systems.
\newblock {\em Automatica}, 109:108533, 2019.

\bibitem{kesten1998algorithmic}
Yonit Kesten, Amir Pnueli, and Lion Raviv.
\newblock Algorithmic verification of linear temporal logic specifications.
\newblock In {\em International Colloquium on Automata, Languages, and
  Programming}, pages 1--16. Springer, 1998.

\bibitem{kushner1967stochastic}
Harold~J Kushner.
\newblock Stochastic stability and control.
\newblock Technical report, Brown Univ Providence RI, 1967.

\bibitem{LAB15}
M.~Lahijanian, S.~B. Andersson, and C.~Belta.
\newblock Formal verification and synthesis for discrete-time stochastic
  systems.
\newblock {\em IEEE Transactions on Automatic Control}, 60(8):2031--2045, Aug
  2015.

\bibitem{majumdar2020symbolic}
Rupak Majumdar, Kaushik Mallik, and Sadegh Soudjani.
\newblock Symbolic controller synthesis for {B}{\"u}chi specifications on
  stochastic systems.
\newblock In {\em Proceedings of the 23rd International Conference on Hybrid
  Systems: Computation and Control}, pages 1--11, 2020.

\bibitem{murali2022scenario}
Vishnu Murali, Ashutosh Trivedi, and Majid Zamani.
\newblock A scenario approach for synthesizing k-inductive barrier
  certificates.
\newblock {\em IEEE Control Systems Letters}, 6:3247--3252, 2022.

\bibitem{nejat2021}
Ameneh Nejati, Abolfazl Lavaei, Pushpak Jagtap, Sadegh Soudjani, and Majid
  Zamani.
\newblock Formal verification of unknown discrete- and continuous-time
  systems:a data-driven approach.
\newblock {\em Under review}, 2021.

\bibitem{niu2021safety}
Luyao Niu, Hongchao Zhang, and Andrew Clark.
\newblock Safety-critical control synthesis for unknown sampled-data systems
  via control barrier functions.
\newblock In {\em 2021 60th IEEE Conference on Decision and Control (CDC)},
  pages 6806--6813. IEEE, 2021.

\bibitem{plambeck2022}
Swantje Plambeck, G{\"o}rschwin Fey, and Schyga.
\newblock Decision tree models of continuous systems.
\newblock In {\em 27th International Conference on Emerging Technologies and
  Factory Automation (ETFA)}. IEEE, 2022.

\bibitem{prajna2004safety}
Stephen Prajna and Ali Jadbabaie.
\newblock Safety verification of hybrid systems using barrier certificates.
\newblock In {\em International Workshop on Hybrid Systems: Computation and
  Control}, pages 477--492. Springer, 2004.

\bibitem{prajna2007framework}
Stephen Prajna, Ali Jadbabaie, and George~J Pappas.
\newblock A framework for worst-case and stochastic safety verification using
  barrier certificates.
\newblock {\em IEEE Transactions on Automatic Control}, 52(8):1415--1428, 2007.

\bibitem{robey2020learning}
Alexander Robey, Haimin Hu, Lars Lindemann, Hanwen Zhang, Dimos~V Dimarogonas,
  Stephen Tu, and Nikolai Matni.
\newblock Learning control barrier functions from expert demonstrations.
\newblock {\em arXiv:2004.03315}, 2020.

\bibitem{robey2021learning}
Alexander Robey, Lars Lindemann, Stephen Tu, and Nikolai Matni.
\newblock Learning robust hybrid control barrier functions for uncertain
  systems.
\newblock {\em IFAC-PapersOnLine}, 54(5):1--6, 2021.

\bibitem{sadraddini2018formal}
Sadra Sadraddini and Calin Belta.
\newblock Formal guarantees in data-driven model identification and control
  synthesis.
\newblock In {\em Proceedings of the 21st International Conference on Hybrid
  Systems: Computation and Control (part of CPS Week)}, pages 147--156, 2018.

\bibitem{salamati2020data}
Ali Salamati, Sadegh Soudjani, and Majid Zamani.
\newblock Data-driven verification under signal temporal logic constraints.
\newblock {\em 21st IFAC World Congress}, 2020.

\bibitem{sloth2012compositional}
Christoffer Sloth, George~J Pappas, and Rafael Wisniewski.
\newblock Compositional safety analysis using barrier certificates.
\newblock In {\em Proceedings of the 15th ACM international conference on
  Hybrid Systems: Computation and Control}, pages 15--24, 2012.

\bibitem{esmaeil2013adaptive}
Sadegh Soudjani and Alessandro Abate.
\newblock Adaptive and sequential gridding procedures for the abstraction and
  verification of stochastic processes.
\newblock {\em SIAM Journal on Applied Dynamical Systems}, 12(2):921--956,
  2013.

\bibitem{soudjani2015dynamic}
Sadegh Soudjani, Alessandro Abate, and Rupak Majumdar.
\newblock Dynamic {B}ayesian networks as formal abstractions of structured
  stochastic processes.
\newblock In {\em 26th International Conference on Concurrency Theory}, pages
  169--183. Schloss Dagstuhl, 2015.

\bibitem{esmaeil2015faust}
Sadegh Soudjani, Caspar Gevaerts, and Alessandro Abate.
\newblock Faust 2: Formal abstractions of uncountable-state stochastic
  processes.
\newblock In {\em 21st International Conference on Tools and Algorithms for the
  Construction and Analysis of Systems (TACAS 2015)}. Newcastle University,
  2015.

\bibitem{SM18_Concentration}
Sadegh Soudjani and Rupak Majumdar.
\newblock Concentration of measure for chance-constrained optimization.
\newblock {\em IFAC-PapersOnLine}, 51(16):277--282, 2018.

\bibitem{SVORENOVA2017230}
M{\'a}ria Svore{\v n}ov{\'a}, Jan K{\v r}et{\'\i}nsk{\'y}, Martin Chmel{\'\i}k,
  Krishnendu Chatterjee, Ivana {\v C}ern{\'a}, and Calin Belta.
\newblock Temporal logic control for stochastic linear systems using
  abstraction refinement of probabilistic games.
\newblock {\em Nonlinear Analysis: Hybrid Systems}, 23:230 -- 253, 2017.

\bibitem{SZ.19}
Abdalla Swikir and Majid Zamani.
\newblock Compositional synthesis of symbolic models for networks of switched
  systems.
\newblock {\em IEEE Control Syst. Lett.}, 3(4):1056--1061, 2019.

\bibitem{tabuada09}
Paulo Tabuada.
\newblock {\em Verification and Control of Hybrid Systems: A Symbolic
  Approach}.
\newblock Springer, 2009.

\bibitem{tabuada2020data}
Paulo Tabuada and Lucas Fraile.
\newblock Data-driven stabilization of {SISO} feedback linearizable systems.
\newblock {\em arXiv preprint arXiv:2003.14240}, 2020.

\bibitem{wang2017safety}
Li~Wang, Aaron~D Ames, and Magnus Egerstedt.
\newblock Safety barrier certificates for collisions-free multirobot systems.
\newblock {\em IEEE Transactions on Robotics}, 33(3):661--674, 2017.

\bibitem{wang2019data}
Zheming Wang and Rapha{\"e}l~M Jungers.
\newblock Data-driven computation of invariant sets of discrete time-invariant
  black-box systems.
\newblock {\em arXiv:1907.12075}, 2019.

\bibitem{wijesuriya2019bayes}
Viraj~Brian Wijesuriya and Alessandro Abate.
\newblock Bayes-adaptive planning for data-efficient verification of uncertain
  {M}arkov decision processes.
\newblock In {\em International Conference on Quantitative Evaluation of
  Systems}, pages 91--108. Springer, 2019.

\bibitem{wood1996estimation}
GR~Wood and BP~Zhang.
\newblock Estimation of the lipschitz constant of a function.
\newblock {\em Journal of Global Optimization}, 8(1):91--103, 1996.

\bibitem{yang2020efficient}
Zhengfeng Yang, Min Wu, and Wang Lin.
\newblock An efficient framework for barrier certificate generation of
  uncertain nonlinear hybrid systems.
\newblock {\em Nonlinear Analysis: Hybrid Systems}, 36:100837, 2020.

\bibitem{Zamani.2017b}
Majid Zamani and Murat Arcak.
\newblock Compositional abstraction for networks of control systems: A
  dissipativity approach.
\newblock {\em IEEE Trans. Control Network Syst.}, 5(3):1003--1015, 2018.

\bibitem{zamani2014symbolic}
Majid Zamani, Peyman~Mohajerin Esfahani, Rupak Majumdar, Alessandro Abate, and
  John Lygeros.
\newblock Symbolic control of stochastic systems via approximately bisimilar
  finite abstractions.
\newblock {\em IEEE Transactions on Automatic Control}, 59(12):3135--3150,
  2014.

\bibitem{zamani2017towards}
Majid Zamani, Ilya Tkachev, and Alessandro Abate.
\newblock Towards scalable synthesis of stochastic control systems.
\newblock {\em Discrete Event Dynamic Systems}, 27(2):341--369, 2017.

\bibitem{zhang2010safety}
Lijun Zhang, Zhikun She, Stefan Ratschan, Holger Hermanns, and Ernst~Moritz
  Hahn.
\newblock Safety verification for probabilistic hybrid systems.
\newblock In {\em International Conference on Computer Aided Verification},
  pages 196--211. Springer, 2010.

\end{thebibliography}

\appendix
\section{Lipschitz continuity of the max function}
%Proof of Lemma \ref{lem:lipseveral}
\Ali{
	\begin{lemma}
		The maximum of Lipschitz continuous functions $f_i:X\rightarrow \mathbb R$, $i=1,2,\ldots,m$, is a Lipschitz continuous function. The Lipschitz constant of the maximum is the sum of the Lipschitz constants of $f_i$.
		\label{lem:lipseveral} 
	\end{lemma}
	%\textcolor{orange}{The proof can be found in Appendix. Now, we raise the following assumptions before introducing the main theorem.}
}
\begin{proof}
	\Ali{Suppose that two Lipschitz continuous functions $f_1$ and $f_2$ have Lipschitz constants $L_1$ and $L_2$, respectively.
		One can rewrite $g = \max(f_1,f_2)$ as:
		\begin{align*}
			g = \max(f_1,f_2)=\frac{f_1+f_2+|f_1-f_2|}{2}.
			%\label{eq:lipsev}
		\end{align*}
		Then, we can use triangle inequality to show that
		\begin{align*}
			|g(x) - g(y)| & \le \frac{1}{2}[|f_1(x)-f_1(y)|+|f_2(x)-f_2(y)| + \\ &  \big ||f_1(x)-f_2(x)|-|f_1(y)-f_2(y)|\big | ]\\
			& \le \frac{1}{2}[L_1\|x-y\| + L_2\|x-y\| + |f_1(x)-f_1(y)| + \\ & |f_2(x)-f_2(y)|] \le \frac{1}{2}[L_1\|x-y\| + L_2\|x-y\| + \\& L_1\|x-y\| + L_2\|x-y\|]
			= (L_1+L_2)\|x-y\|.
		\end{align*}
		%so we just need to show that $|f_1-f_2|$ is Lipschitz continuous.
		%One has $|f_1(x)-f_2(x)-f_1(y)+f_2(y)|\leq |f_1(x)-f_1(y)|+|f_2(x)-f_2(y)|\leq L_1 \|x-y\|+L_2\|x-y\|=(L_1+L_2)\|x-y\|$. 
		Therefore, $\max(f_1,f_2)$ is also a Lipschitz continuous function with Lipschitz constant $L_1+L_2$. 
		This argument can be extended inductively to the maximum of every number of functions.
	}
\end{proof}
\begin{lemma}
	\Ali{For any two analytic functions $f_1:X\rightarrow \mathbb R$ and $f_2:X\rightarrow \mathbb R$ with a compact domain $X$, $L:=\max(L_1,L_2)$ is a Lipschitz constant of $\max(f_1,f_2)$.}
\end{lemma}
\begin{proof}
	\Ali{
		Note that
		\begin{equation*}
			g(x) = \max(f_1(x),f_2(x)) = \begin{cases}
				f_1(x) & \text{if} \quad f_1(x) - f_2(x)\ge 0\\
				f_2(x) & \text{if} \quad f_1(x)- f_2(x)\le 0.
			\end{cases}
		\end{equation*}
		The function $f_1 - f_2$ is also analytic, thus has a finite number of zeros in a compact domain. Let us denote the finite set of zeros as $Z$.
		We first show this for one-dimensional compact domains $X\subset \mathbb R$.
		%We show that the Lipschitz constant of $\max(f_1,f_2)$ can be further improved to $\max(L_1,L_2)$.
		Take two points $x,y\in X$ such that $x<y$, and define $Z\cap [x,y] = \{z_1,z_2,\ldots,z_m\}$
		such that $z_i<z_{i+1}$ for any $i=1,2,\ldots,m-1$.
		Then we have
		\begin{align*}
			|g(y) - g(x)| = |f_{i_y}(y) - f_{i_m}(z_m) & + f_{i_m}(z_m) - f_{i_{m-1}}(z_{m-1})+\ldots\\
			& + f_{i_2}(z_2)-f_{i_1}(z_1) + f_{i_1}(z_1) - f_{i_x}(x)|,
		\end{align*}
		for some appropriate choices of $i_x,i_y,i_1,\ldots, i_m$ all from the set $\{1,2\}$. Since $g(z_j) = f_1(z_j) - f_2(z_j) = 0$, we can set the index of $f$ to symbol that belongs to the set $\{1,2\}$ when the function is evaluated at any $z_j$. Then, we have
		\begin{align*}
			& |g(y) - g(x)|\\
			& = |f_{i_y}(y) - f_{i_y}(z_m) + f_{i_m}(z_m) - f_{i_m}(z_{m-1})+\ldots + \\ & f_{i_2}(z_2)-f_{i_2}(z_1) + f_{i_x}(z_1) - f_{i_x}(x)|
			\le \\ & |f_{i_y}(y) - f_{i_y}(z_m)| + |f_{i_m}(z_m) - f_{i_m}(z_{m-1})|+\ldots +  \\ & |f_{i_2}(z_2)-f_{i_2}(z_1) | + |f_{i_x}(z_1) - f_{i_x}(x)|\\
			& \le L_{i_y}(y-z_m) + L_{i_m}(z_m-z_{m-1}) +\ldots +\\ & L_{i_2}(z_2-z_1) + L_{i_x}(z_1-x)\\
			& L(y-z_m) + L(z_m-z_{m-1}) +\ldots +L(z_2-z_1) + L(z_1-x)\\
			&  = L (y-x),
		\end{align*}
		where $L =\max(L_1,L_2) = \max(L_{i_y},L_{i_x},L_{i_1},\ldots,L_{i_m})$. This concludes the proof for one-dimensional case.}
	
	\Ali{We now prove the statement for multi-dimensional case.
		Take two points $x,y\in X\subset\mathbb R^n$ with $x = (x_1,\ldots,x_n)$ and $y = (y_1,\ldots,y_n)$. The functions $f_1,f_2$ have Lipschitz constants $L_1,L_2$, which means
		\begin{equation}
			\label{eq:Lip}
			|f_i(y_1,\ldots,y_n) - f_i(x_1,\ldots,x_n) |\le L_i\|(y_1-x_1,\ldots,y_n-x_n)\|,\quad i\in\{1,2\}.
		\end{equation}
		Define the line segment that connects these two points as $D := \{\lambda y + (1-\lambda)x\,|\, \lambda\in[0,1]\}$. Let us know restrict the domain of the function $g$ to $D$ and define:
		\begin{align*}
			& h:[0,1]\rightarrow\mathbb R,\quad h(\lambda) := g(\lambda y + (1-\lambda)x) = \\ & \max(f_1(\lambda y + (1-\lambda)x),f_2(\lambda y + (1-\lambda)x)).
		\end{align*}
		We can now apply the first part of the proof to get:
		\begin{equation}
			\label{eq:h}
			|h(1)-h(0)|\le L'|1-0|,
		\end{equation}
		where $L'$ is the maximum of the Lipschitz constants of $f_1(\lambda y + (1-\lambda)x)$ and $f_2(\lambda y + (1-\lambda)x)$ with respect to $\lambda$. To get these Lipschitz constants, we use \eqref{eq:Lip}:
		\begin{align*}
			& |f_i(\lambda_1 y + (1-\lambda_1)x) - f_i(\lambda_2 y + (1-\lambda_2)x)|
			\le \\ & L_i\|(\lambda_1-\lambda_2)(y-x)\| = L_i|\lambda_1-\lambda_2|\,\|y-x\| \\
			& = 
			(L_i\|y-x\|)\,|\lambda_1-\lambda_2|
		\end{align*}
		Therefore, the Lipschitz constants of $f_1(\lambda y + (1-\lambda)x)$ for a given $x,y$ with respect to $\lambda$ is $L_i\|y-x\|$.
		Replacing definitions in \eqref{eq:h}, we have
		\begin{equation*}
			|g(y)-g(x)|\le L' = \max(L_1\|y-x\|,L_2\|y-x\|) =\|y-x\|\max(L_1,L_2).
		\end{equation*}
		This completes the proof.}
\end{proof}
\Ali{\section{Proof of Corollary \ref{rem:gofepsilon}}
	\label{sec:proof_g}
	The probability distribution from which $x_i$ is sampled must satisfy Assumption~\ref{ass:G}. This assumption requires having a strictly increasing function $G:\mathbb{R}_0^+ \rightarrow [0,1]$ that satisfies
	\begin{align*}
		\pr[\mathrm{b}(x,r)]\geq G(r),\qquad \forall x \in X.
		%\label{eq:gofr}
	\end{align*}
	Since we assume that samples are collected uniformly, $\pr[\mathrm{b}(x,r)]$ for every small ball centered at every $x \in X$ with radius $r=\epsilon$ can be computed by dividing the volume of this ball by the whole state set volume. Given that one needs to find the maximum ball that is valid for $\forall x \in X$, and some points $x$ lie on the border of the hyper-rectangular state set, the maximum ball is a semi-hypersphere in general, whose volume can be computed as $\dfrac{1}{2^n}\frac{\pi^{\frac{n}{2}}}{\Gamma(\frac{n}{2}+1)}\epsilon^n$ with the Gamma function defined as $\Gamma(k) = 1\times 2\times 3\ldots\times(k-1)$ and $\Gamma(k+\frac{1}{2})=\frac{1}{2}\times \frac{3}{2}\times\ldots(k-\frac{3}{2})(k-\frac{1}{2})\pi^{\frac{1}{2}}$ for all positive integers. Dividing this value by the whole state set volume, which is $\prod_{i=1}^{n} \eta_x(i)$ for $\eta_x(i)$ as the length of the edges in each direction, gives us $G(\epsilon)$.}

\Ali{\section{Proof of Corollary \ref{rem:gofepsilon2}}
	The proof is similar to the proof of Corollary~\ref{rem:gofepsilon} in ~\ref{sec:proof_g}. Here, the centered ball with the maximum volume is the intersection of the whole state set sphere and the small ball $r=\epsilon$ centered at any point on the border of the state set sphere. The volume of this intersection, which is the volume of two separate caps, can be computed as:
	\begin{align*}
		V_n^{cap}(\tilde r, \rm c_1) + V_n^{cap}(\epsilon , \rm c_2),
	\end{align*}
	where 
	\begin{align*}
		V_n^{cap}(\tilde r, \rm c_1) = \frac{1}{2}\frac{\pi ^{\frac{n}{2}}}{\Gamma(\frac{n}{2}+1)}\tilde r^n \mathrm I(1-\frac{\rm c_1^2}{\tilde r^2};\frac{n+1}{2},\frac{1}{2}),
	\end{align*}
	and
	\begin{align*}
		V_n^{cap}(\epsilon, \rm c_2) = \frac{1}{2}\frac{\pi ^{\frac{n}{2}}}{\Gamma(\frac{n}{2}+1)}\epsilon^n \mathrm I(1-\frac{\rm c_2^2}{\epsilon^2};\frac{n+1}{2},\frac{1}{2}),
	\end{align*}
	for $\rm c_1 = \frac{2 \tilde r^2- \epsilon ^2}{2 \tilde r}$, and $\rm c_2 = \frac{\epsilon^2}{2 \tilde r}$.
	By dividing the intersection volume by the volume of the whole hypersphere state set, which is 
	\begin{align*}
		V_n(\tilde r) =\frac{\pi^{\frac{n}{2}}}{\Gamma(\frac{n}{2}+1)}\tilde r^n,
	\end{align*}
	one can compute $G(\epsilon)$ as in Corollary \ref{rem:gofepsilon2}.
}

\Ali{
	\section{Coefficients of the computed barrier certificates in floating point format with 16 digits.}
	\begin{center}
		\resizebox{.7\columnwidth}{!}{%
			\begin{tabular}{||c c c ||} 
				\hline
				Temperature Verification& Lane Keeping & Synthesizing a \\
				for 3 Rooms&  System & Controller \\
				[0.5ex] 
				\hline
				$1.118824712343290 \times 10^{-1}$ & $2.200050812923097\times 10^{-4}$ & $1.189325015407815 \times 10$ \\
				\hline
				$1.121295401333170 \times 10^{-1}$ & $3.901846347425760 \times 10^{-1}$ & $-1.070392322770013\times10^3$  \\
				\hline
				$1.122576531449860 \times 10^{-1}$ & $1.480240596483330 \times 10^{-1}$ & $3.612276124685787\times10^4$  \\
				\hline
				$-3.751401155407000 \times 10^{-3}$ & $-2.825312554914731\times 10^{-4}$ & $-5.417521260597183\times10^5$\\ 
				\hline
				$-4.728480781000000 \times 10^{-3}$ & $9.905388481691000 \times 10^{-3}$ & $3.046603167514221\times 10^6$\\ 
				\hline
				$-2.284303936564000 \times 10^{-3}$ & $-6.672383448890000 \times 10^{-3}$ & -\\ 
				\hline
				$-3.761231117922648 \times 10^0$ & $-6.918249590565419\times 10^{-4}$ & -\\ 
				\hline
				$-3.815332731044874 \times 10^0$ & $4.678025224577894\times 10^{-4}$ & -\\ 
				\hline
				$-3.803570830339135 \times 10^0$ & $-1.539512818952500 \times 10^{-2}$ & -\\ 
				\hline
				$9.993049903406006 \times 10$ & $4.518033593474370 \times 10^{-1}$ & -\\ 
				\hline
			\end{tabular}
		}
	\end{center}
}
\Ali{In the above table, the values in first two columns from top to the bottom are $\{b_0,\ldots,b_9\}$ in respective case studies. The values in the the third column from top to the bottom are $\{b_0,\ldots,b_4\}$ in the last case study.}
\end{document}